\documentclass[11pt, oneside]{amsart}   	
\usepackage{geometry}                		
\usepackage{graphicx}				

\newcommand{\R}{\mathbb{R}}

\usepackage{amssymb,amsmath,amsfonts,latexsym}
\usepackage{bbm}
\newtheorem{theorem}{Theorem}

\title{A Note on Markov Normalized Magnetic Eigenmaps}
\author{Alexander Cloninger}
\date{}							

\begin{document}
\maketitle

\begin{abstract}
We note that building a magnetic Laplacian from the Markov transition matrix, rather than the graph adjacency matrix, yields several benefits for the magnetic eigenmaps algorithm. The two largest benefits are that the embedding becomes more stable as a function of the rotation parameter g, and the principal eigenvector of the magnetic Laplacian now converges to the page rank of the network as a function of diffusion time. We show empirically that this normalization improves the phase and real/imaginary embeddings of the low-frequency eigenvectors of the magnetic Laplacian.
\end{abstract}

\section{Introduction}
We consider the problem of building a diffusion operator on directed networks and graphs with adjacency matrix $W\in \R^{N\times N}$ that parallels Diffusion Maps on undirected networks \cite{diffMaps}.  A natural approach to this is to build the Magnetic Eigenmaps Laplacian \cite{magMaps},
\begin{eqnarray*}
L &=& D - T\circ W^{(s)},
\end{eqnarray*}
where $(T\circ W^{(s)})_{i,j} = e^{2\pi \mathrm{i} g (W_{j,i} - W_{i,j})} \frac{W_{i,j} + W_{j,i}}{2}$ and $D_{i,i} = \sum_j W^{(s)}_{i,j}$.  This method assigns a complex rotation to asymmetric links.  This approach generalizes the notion of building a diffusion map or graph embedding to work on directed graphs.

In this note, we observe that Markov normalization of the adjacency matrix,
\begin{eqnarray*}
P(x,y) = \frac{W(x,y)}{\sum_y W(x,y)},
\end{eqnarray*} 
allows us to introduce the notion of diffusion time steps, which is capable of capturing the evolution of the process.   Markov normalization of the graph adjacency matrix yields several benefits:
\begin{enumerate}
\item normalizing the complex rotation matrix $e^{2\pi\mathrm{i} g (P_{j,i} - P_{i,j})}$ by the density of the node, which emphasizes absorbing states and converges to the stationary distribution of the process,
\item stabilizing the phase embedding with respect to choice of the rotation parameter $g$,
\item separating the long time trend of the process into the first eigenvector and recurrent states into subsequent eigenvectors, and
\item introducing diffusion time, which allows for the study of multi-step neighborhoods.
\end{enumerate}

We further note that it is possible to augment the adjacency matrix by a transportation factor, as is done in the PageRank algorithm , when there exists an absorbing state \cite{pagerank}.  For a small $\alpha>0$, we replace $P \rightarrow (1-\alpha)P + \alpha \mathbbm{1} \mathbbm{1}^{\intercal}$.  This creates a small probability of jumping from any node to any other node, and thus escaping a sink in the process and making the process ergodic.

For the rest of the paper we will assume that the adjacency matrix $P$ is ergodic, under the pretense that if there is an absorbing state, we simply add a transportation factor.

\section{Markov Property and Diffusion Time}

An important aspect of diffusion maps is the ability to separate the statistics of the sampling distribution from the geometry of the underlying manifold \cite{lafonThesis}.  This is accomplished by normalizing the affinity between points by the overall transition probability,
\begin{eqnarray*}
P = T^{-1} W, & \textnormal{where} & T_{x,y} = \begin{cases} \sum_y W(x,y), & x=y \\ 0, & x\neq y \end{cases}.
\end{eqnarray*}

This also allows us to consider the properties of the multiple time step diffusion process $P^t$ for $t \in \R^+$.  For a symmetric process (i.e. undirected graph), the embedding of this process remains stable.  This is because, if the eigendecomposition is $P = \Phi \Lambda \Phi^*$, then
\begin{eqnarray*}
P^t = (\Phi \Lambda \Phi^*)^t =  \Phi  \Lambda^t \Phi^*.
\end{eqnarray*}
However, for a non-symmetric graph, this property no longer holds.  Thus the eigenvectors of $P^t$ will exhibit time dependence beyond a simple multiplicity factor.  

The markov normalized magnetic Laplacian then takes the form of
\begin{eqnarray*}
L^{(t)}_{i,j} = D_{i,j} - e^{2\pi \mathrm{i} g (P^t_{j,i} - P^t_{i,j})} \frac{P^t_{i,j} + P^t_{j,i}}{2}, &\textnormal{where}& D_{i,i} = \frac{1}{2}\sum_j \left(P^t_{i,j} + P^t_{j,i}\right).
\end{eqnarray*}

Adjusting the weight matrix to be Markov normalized has several implications on the magnetic Laplacian.  In diffusion on a symmetric graph, normalizing by the node density causes the principal eigenvector $\phi_0$ to be a constant across all nodes and is thus trivial.  However, when the network is not symmetric, the principal eigenvector takes on a much more important role.  Namely, the principal eigenvector recovers the degree of asymmetry between nodes.  This implies that, as $t$ increases, the phase of the eigenvectors converges to the page rank of the network.

\begin{theorem}
Suppose there exists $T$ such that $P^T_{i,j}>0$ for all $i,j$.  Let $L_{norm} = D^{-1/2}LD^{-1/2}$ be the normalized magnetic Laplacian.  Then in the limit as $t\rightarrow \infty$, the principal eigenvector of the normalized Laplacian $L_{norm} = \lim_t L_{norm}^{(t)}$ (denoted $\phi$) converge to 
\begin{eqnarray*}
\phi_{i}^{(g)} = c \cdot e^{2\pi g h_i} \left(\frac{1+\sum_j P_{j,i}}{2}\right)^{1/2},
\end{eqnarray*}
where $h = pagerank(P)$ (i.e. the stationary distribution of $P$), and $c\in \mathbb{C}$.
\end{theorem}
\begin{proof}
From Theorem 1 in \cite{magMaps}, we know that the magnetic Laplacian $L$ has a zero eigenvalue iff $\exists h$ such that $a_{i,j} = h_j - h_i$.  Furthermore, we know that, in this case, $\phi^{(g)}_{k,i} = e^{2\pi \mathrm{i} g h_i} \phi^{(0)}_{k,i}$.   We also know that, since $P$ is ergodic by assumption, $P^t_{i,\cdot} = h$ as $t\rightarrow \infty$, where $h = pagerank(P)$.  That means $a_{i,j} = P^t_{i,j}-P^t_{j,i} = h_j - h_i$. 

Let $v$ be the principal eigenvector of $I - D^{-1/2}P_{sym} D^{-1/2}$, where $P_{sym} = \left(\frac{P+P^*}{2}\right)$ and $D_{xx} = \sum_y (P_{sym})_{xy} = \frac{1+\sum_x P_{xy}}{2}$.  Then
\begin{eqnarray*}
\left(I - D^{-1/2}P_{sym} D^{-1/2}\right) v = 0 &\implies& \left(D^{1/2} - P_{sym} D^{-1/2}\right) v = 0\\
&\implies& \left(D - P_{sym} \right) x = 0\\
&\implies& \left(I - D^{-1} P_{sym} \right)x = 0,
\end{eqnarray*}
where $x = D^{-1/2}v$.  Because $D^{-1} P_{sym}$ is row stochastic, we know $x = \mathbbm{1}$, so $v = D^{1/2}\mathbbm{1}$.

Since $\phi^{(g)}_{k,i} = e^{2\pi \mathrm{i} g h_i} \phi^{(0)}_{k,i}$, we know 
\begin{eqnarray*}
\phi^{(g)}_{i} &=& e^{2\pi \mathrm{i} g h_i} (c\cdot D^{1/2}_i )\\
&=& c\cdot e^{2\pi \mathrm{i} g h_i} \left(\frac{1+\sum_j P_{j,i}}{2}\right)^{1/2}.
\end{eqnarray*}

\end{proof}

This only applies in the limit as $t\rightarrow \infty$.  The embeddings of $L^{(t)}$ for small $t$, and even $t=1$, yield valuable properties due to the Markov normalization.  We show this empirically in Section \ref{examples}.


\section{Examples}\label{examples}
In all examples of the Markov normalized process, we will refer to a value of $g\in[0,1/2)$ that is on the same scale as for the unnormalized Laplacian.  However, due to the weights being normalized to sum to 1, they are actually on a different scale than the unnormalized weights matrix.  For that reason, we set $g \leftarrow g / max(max(P))$ for the Markov normalized weights so that the degree of rotation is on the same order of magnitude for both methods.

\subsection{Three Cluster Example}\label{threecluster}
We begin with the three cluster example from \cite{magMaps}.  This example is particularly simple given that the complex rotation of three clusters will, more often than not, keep the clusters separated.  However, for small values of $g$, rotation of the unnormalized asymmetric weights matrix yields insufficient separation between the clusters, as shown in Figure \ref{fig:theeClusters}.   In fact, Figure \ref{fig:threeClustersRandG} shows the percentage of points that cluster correctly across 100 uniform random small values of $g$ ($g<0.25$), where clustering was done by k-means with $k=3$.

\begin{figure}[!h]
\footnotesize
\begin{tabular}{ccc}
\includegraphics[width=.25\textwidth]{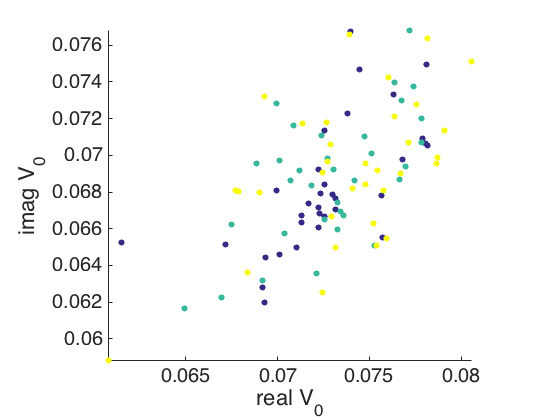} & 
\includegraphics[width=.25\textwidth]{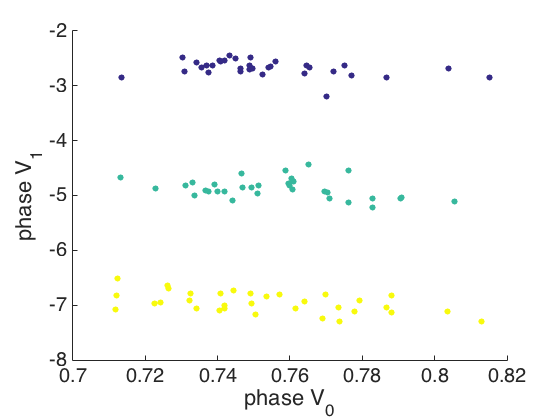} &
\includegraphics[width=.25\textwidth]{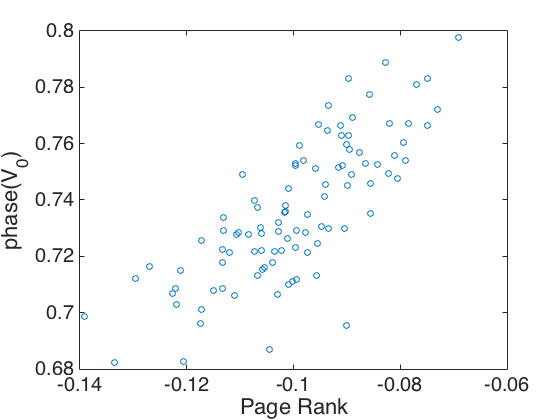} \\
Unnormalized Embedding & Unnormalized Phase & Unnormalized Phase vs Page Rank \\
\includegraphics[width=.25\textwidth]{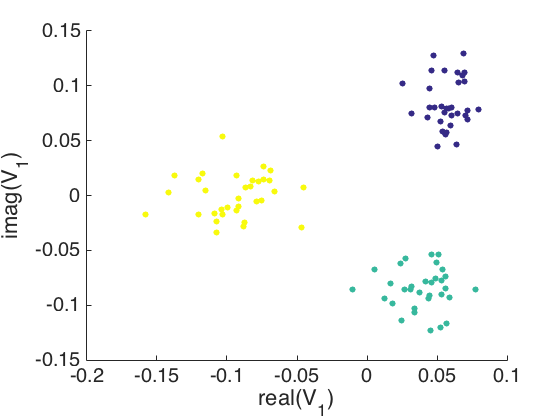} & 
\includegraphics[width=.25\textwidth]{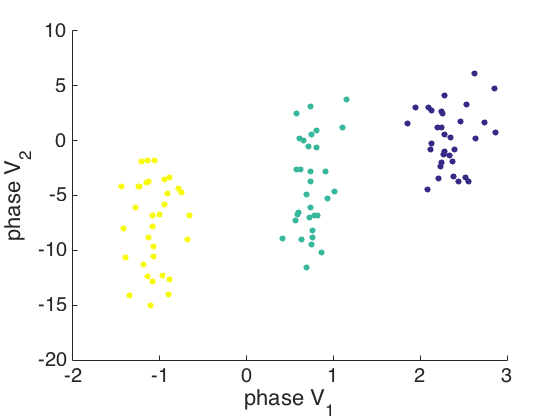} &
\includegraphics[width=.25\textwidth]{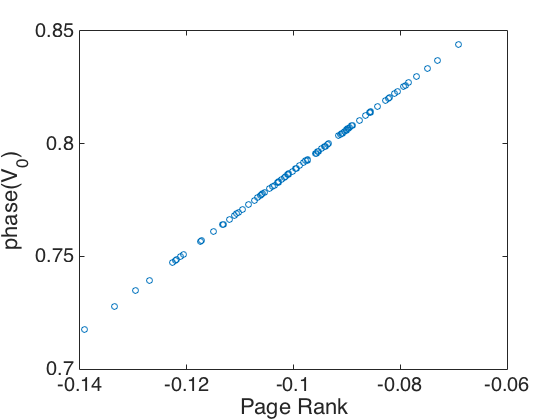} \\
Markov Normalized Embedding & Markov Normalized Phase & Markov Normalized Phase vs Page Rank
\end{tabular}
\caption{Three clusters with $P(inCluster)=0.5$, $P(outCluster)=0.5$, $P(rotateClockwise) = 0.9$, and with parameter $g=0.04$.  The markov normalized embeddings are with $t=1$, though the bottom right page rank plot is for $t=4$ to allow sufficient time to converge. }\label{fig:theeClusters}
\end{figure}

\begin{figure}[!h]
\begin{tabular}{cc}
\includegraphics[width=.45\textwidth]{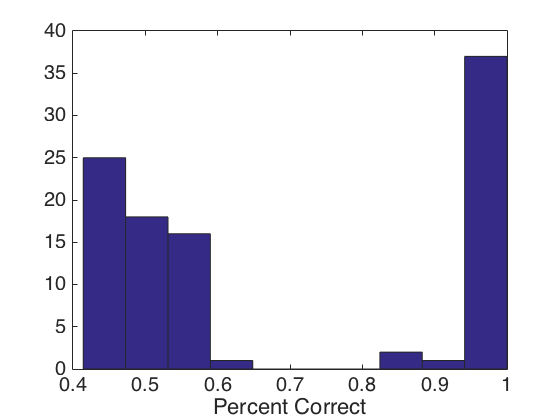} & 
\includegraphics[width=.45\textwidth]{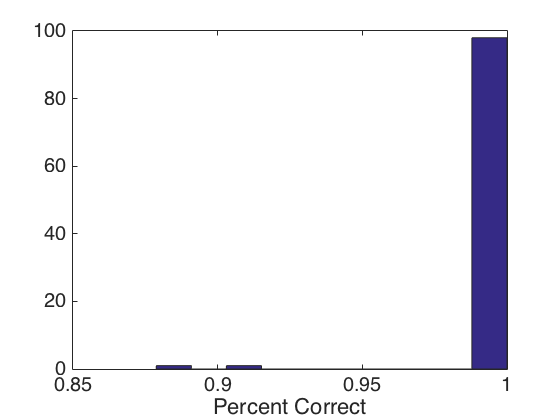} \\
Unnormalized Embedding Hist & Markov Normalized Embedding Hist
\end{tabular}
\caption{Three clusters with $P(inCluster)=0.5$, $P(outCluster)=0.5$, $P(rotateClockwise) = 0.9$, and with random parameter $g<0.25$.  The markov normalized embeddings are with $t=1$. }\label{fig:threeClustersRandG}
\end{figure}

By incorporating the notion of diffusion time, we are also able to observe the time dynamics of this process.  Figure \ref{fig:theeClusterSpin} in the appendix shows the embedding for each $t\in \{1, 2, ..., 9\}$.  Notice that after 3 and 4 time steps, the clusters become momentarily blurred.  This is because, by that time, most of the energy from cluster 1 will have moved to clusters 2 and 3, etc.  However, after that momentary blurring, the process stabilizes and we continue to see the three clusters rotating about each other for larger $t$.  This tells us that there is a global mixing time of $t>1$ for the process to satisfy $P_{i,j} > \epsilon$ for all $i,j$.

\subsection{Circle with Drift}
We consider data distributed on a unit circle, with drift in the counterclockwise direction.  This is an interesting process as it has a very large mixing time.  We define the affinity between points to be 
\begin{eqnarray*}
k(x,y) = \begin{cases} e^{-\|x-y\|^2/5\sigma^2}, &  \measuredangle xy \ge 0 \\  e^{-\|x-y\|^2/\sigma^2}, &  \measuredangle xy < 0\end{cases},
\end{eqnarray*}
where $ \measuredangle xy$ measures the angle between the points.  The affinity matrix is seen in Figure \ref{fig:circleDriftAff} in the appendix.   

Figure \ref{fig:circleDrift} compares the different normalizations and eigenvectors for the circle with drift.  By normalizing the weights matrix, the first eigenvector is once again the trivial eigenvector.  Subsequent eigenvectors recover the geometry of the circle.

\begin{figure}[!h]
\begin{tabular}{cc}
\includegraphics[width=.3\textwidth]{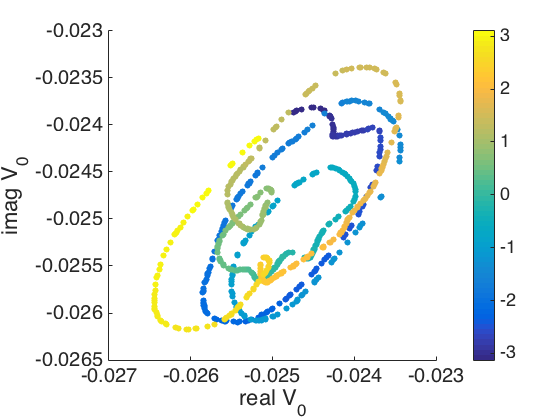} & 
\includegraphics[width=.3\textwidth]{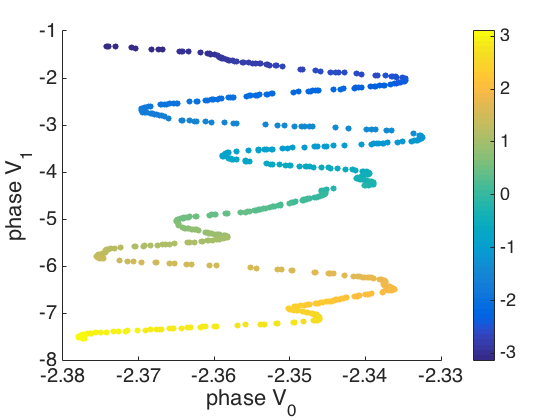} \\
Unnormalized Embedding & Unnormalized Phase \\
\includegraphics[width=.3\textwidth]{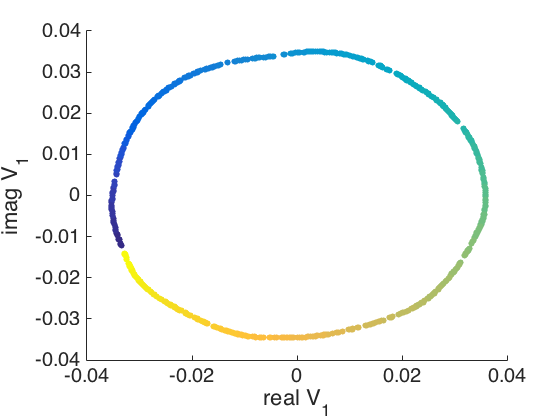} & 
\includegraphics[width=.3\textwidth]{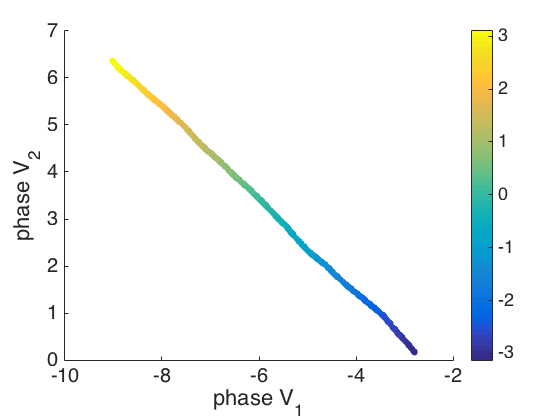} \\
Markov Normalized Embedding & Markov Normalized Phase
\end{tabular}
\caption{Circle with drift, colored by angle in original data.  Top: Unnormalized and leading two eigenvectors, Bottom: Markov normalized and first two non-trivial eigenvectors.  $g=0.04$}\label{fig:circleDrift}
\end{figure}


The higher order eigenvectors also contain useful information, as in diffusion maps.  They mimic the sins and cosines of the angle around the circle, as the Fourier functions $\{e^{\mathrm{i}nx}\}$ are eigenfunctions of the Laplacian on the circle.  We wish to recover these functions independent of the sampling density on the circle.  Figure \ref{fig:circleDriftEigs} in the appendix shows that the Markov normalized weights matrix gives a closer approximation of the sin multiples on the circle that comes from the unnormalized weights matrix.  Note that the data used here is not uniformly sampled across the circle.

\subsection{Bow Tie Example}
Another issue with not including a diffusion time is the need for $g$ to implicitly assume the number of clusters.  By setting $g$ to be $1/k$ for $k$ clusters as the authors suggest \cite{magMaps}, the complex rotation can lose valuable information about transitions that aren't $k$ periodic.  We demonstrate this with an example similar to the three cluster example from Section \ref{threecluster}.  We build a three cluster cycle among clusters 1,2, and 3, and a 5 cluster cycle among clusters 1, 4, 5, 6, and 7.  By normalizing the weights, we see in Figure \ref{fig:bowtie} that the non-principal eigenvectors capture the geometry of the process.  Moreover, we discount the fact that cluster 1 appears in both cycles, and will end with twice the energy of the other clusters for large $t$.
\begin{figure}[!h]
\begin{tabular}{cc}
\includegraphics[width=.3\textwidth]{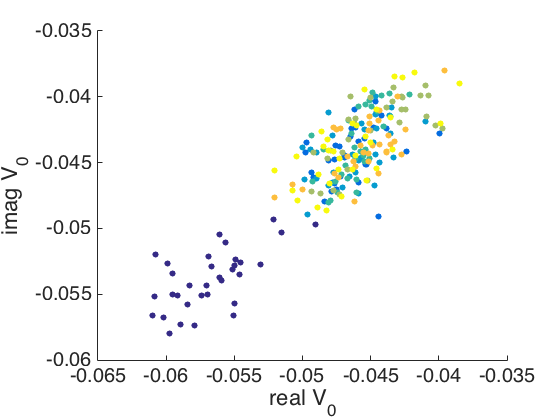} & 
\includegraphics[width=.3\textwidth]{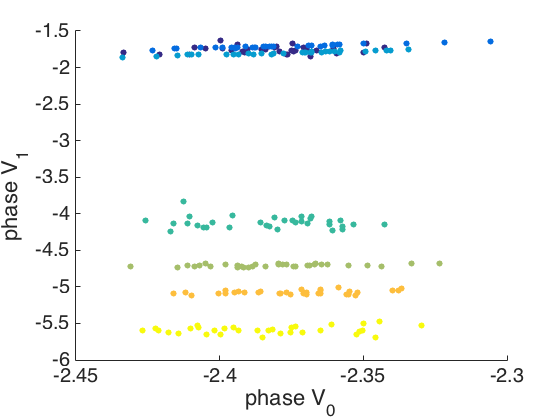} \\
Unnormalized Embedding & Unnormalized Phase \\
\includegraphics[width=.3\textwidth]{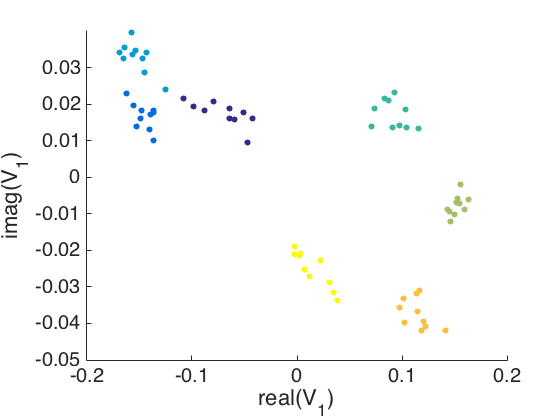} & 
\includegraphics[width=.3\textwidth]{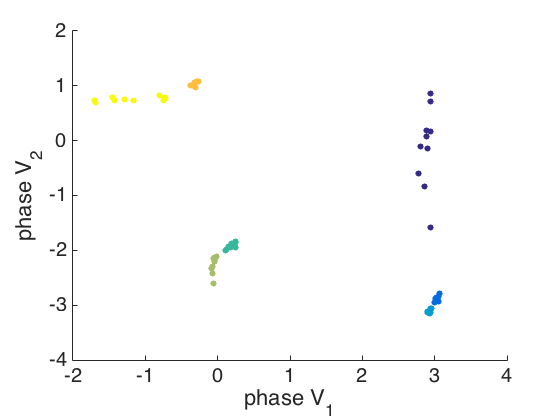} \\
Markov Normalized Embedding & Markov Normalized Phase
\end{tabular}
\caption{Bow tie clusters with $P(inCluster)=0.5$, $P(outCluster)=0.5$, $P(rotateClockwise) = 0.9$, and with $g=0.04$.  For Markov normalization, we take $t=1$.}\label{fig:bowtie}
\end{figure}

We are also able to observe the time dynamics of this process.  Figure \ref{fig:bowtieSpin} in the appendix shows the embedding for each $t\in \{1, 2, ..., 9\}$.  Notice that after 7 and 8 time steps, the clusters become momentarily blurred.  This is because, by that time, most of the energy from one end of the bow tie has moved to the other.  However, after that momentary blurring, the process stabilizes and we continue to see the three clusters rotating about each other for larger $t$.  This tells us that there is a global mixing time of $t>6$ for the process to satisfy $P_{i,j} > \epsilon$ for all $i,j$.  We can also see these dynamics in the affinity matrix in Figure \ref{fig:bowtieAff}.  And finally, we see in Figure \ref{fig:bowtieAff} that after a large $t$, the phase of the leading eigenvector converges to the page rank of the process.

\subsection{Hidden Circle}\label{hiddenCircle}
As a different example, we consider a mix of absorbing and recurrent states.  The data we consider is a unit square with drift from left to right via
\begin{eqnarray*}
k(x,y) = \begin{cases} e^{-\|x-y\|^2/3\sigma^2}, &  x_1 < y_1 \\  e^{-\|x-y\|^2/\sigma^2}, &  x_1 \ge y_1 \end{cases}.
\end{eqnarray*}
On top of that, we add an annulus in the middle of the unit square with counterclockwise flow.  The results of the phase and embedding are in Figure \ref{fig:squareCircle}.

\begin{figure}[!h]
\begin{tabular}{ccc}
\includegraphics[width=.3\textwidth]{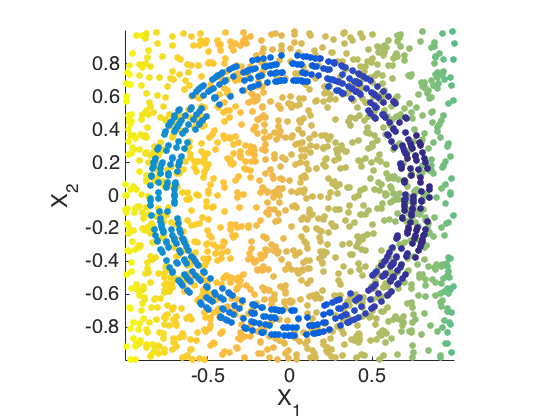} & 
\includegraphics[width=.3\textwidth]{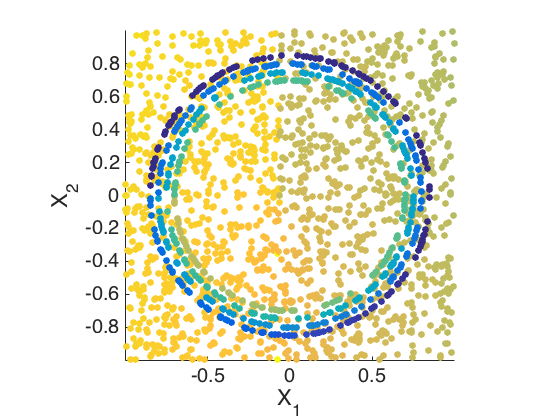} & 
\includegraphics[width=.3\textwidth]{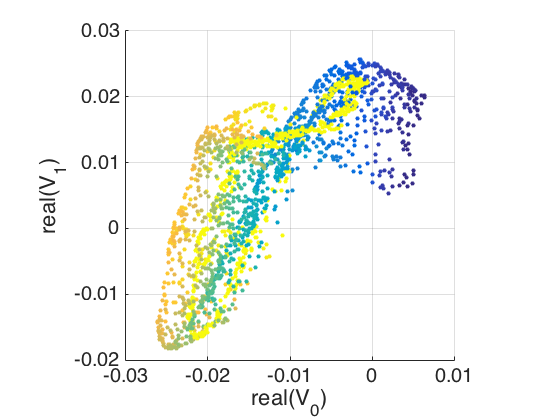} \\
Unnormalized Phase $V_0$ & Unnormalized Phase $V_1$ & Unnormalized Embedding\\
\includegraphics[width=.3\textwidth]{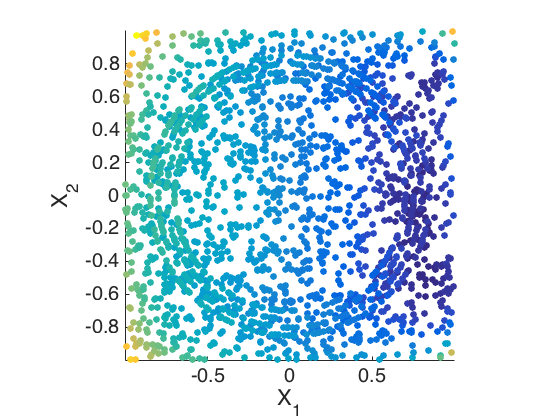} & 
\includegraphics[width=.3\textwidth]{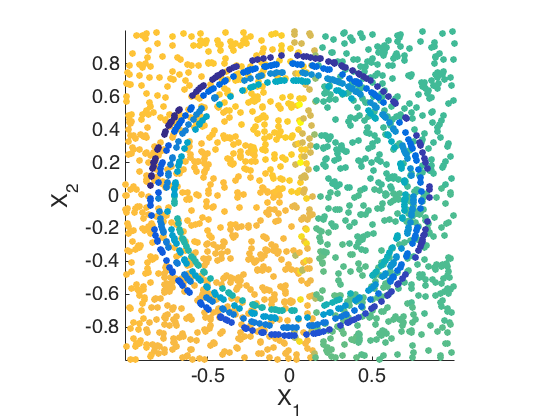} & 
\includegraphics[width=.3\textwidth]{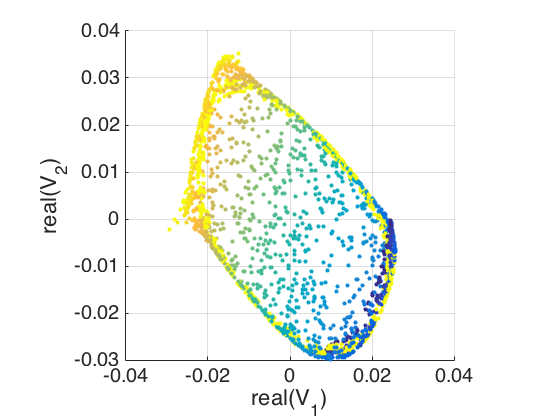} \\
Markov Normalized Phase $V_0$ & Markov Normalized  Phase $V_1$ & Markov Normalized  Embedding\\
\end{tabular}
\caption{$g=0.24$, $t=4$}\label{fig:squareCircle}
\end{figure}

Also, by not normalizing the phase rotation, low density points become separated in the phase space (see Figure \ref{fig:squareCircleTorus}).  This leads to discontinuity in the torus projection of the phase, because the top and bottom edges of the square do not have much energy coming in.  
\begin{figure}[!h]
\begin{tabular}{cc}
\includegraphics[width=.3\textwidth]{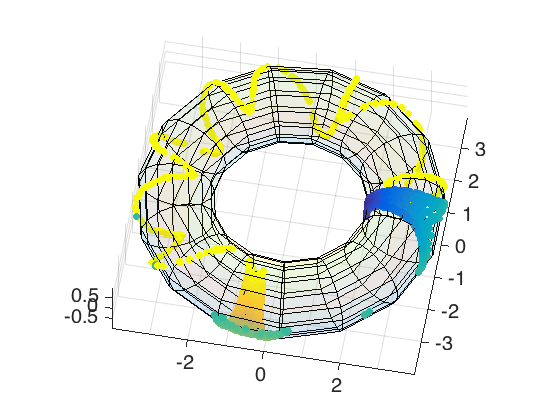} & 
\includegraphics[width=.3\textwidth]{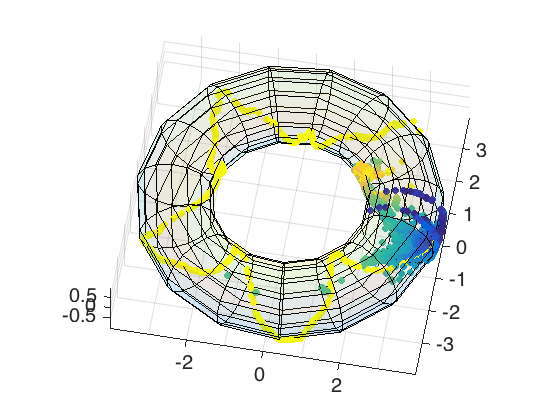} \\
Unnormalized Phase Torus & Markov Normalized Phase Torus\\
\end{tabular}
\caption{$g=0.24$, $t=1$}\label{fig:squareCircleTorus}
\end{figure}

\subsection{Absorbing States}\label{absorbingState}
As a final example, we return to the three cluster example from Section \ref{threecluster}.  We change this by removing one node's out going edges, making that node an absorbing state.  We augment the weights matrix for the normalized weights matrix by adding a teleportation parameter of $\alpha=0.1$.

\begin{figure}[!h]
\begin{tabular}{cc}
\includegraphics[width=.3\textwidth]{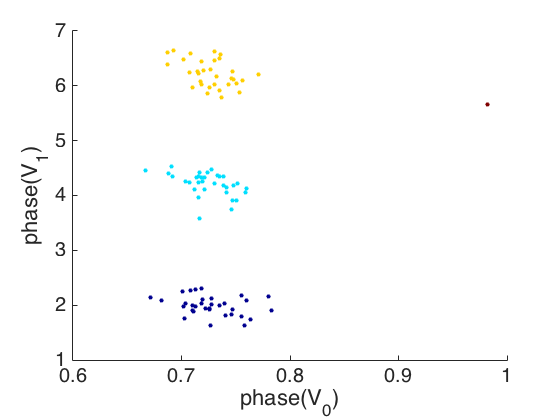} & 
\includegraphics[width=.3\textwidth]{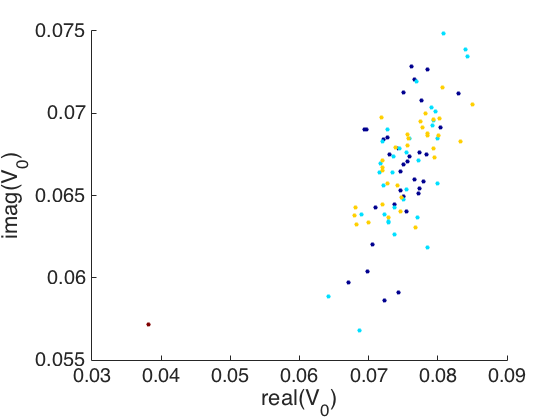} \\
Unnormalized Phase & Unnormalized Embedding \\
\includegraphics[width=.3\textwidth]{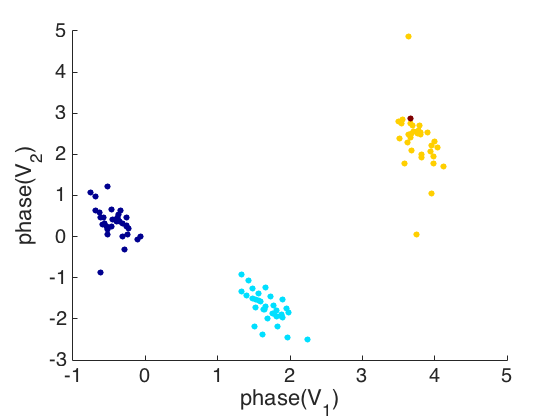} & 
\includegraphics[width=.3\textwidth]{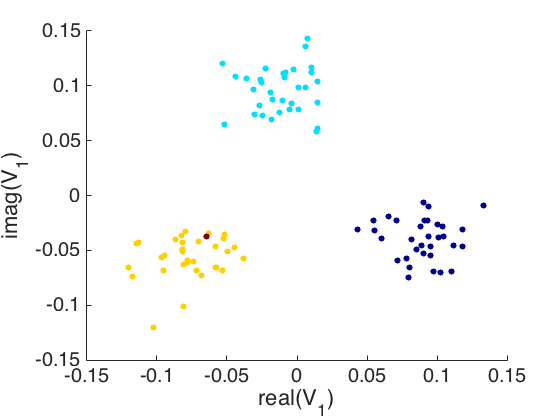} \\
Normalized Phase $t=1$ & Normalized Embedding $t=1$ \\
\includegraphics[width=.3\textwidth]{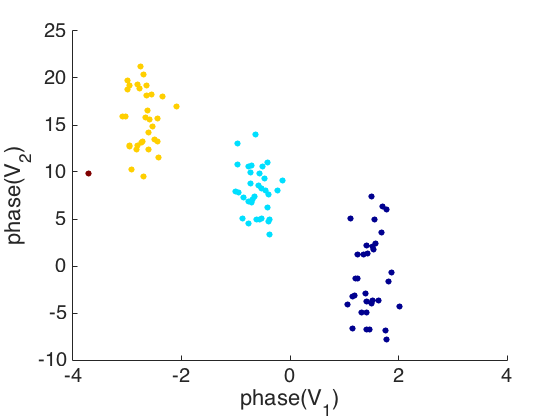} & 
\includegraphics[width=.3\textwidth]{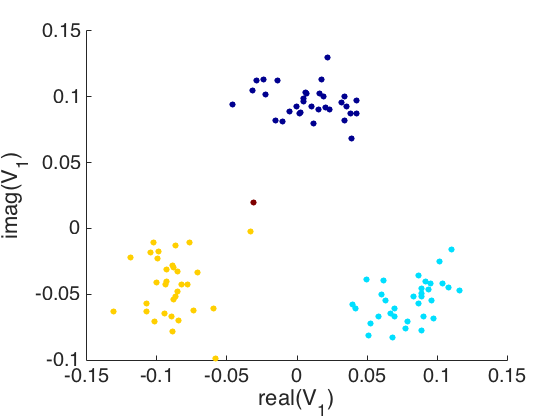} \\
Normalized Phase $t=5$ & Normalized Embedding $t=5$ 
\end{tabular}
\caption{Three cluster data with absorbing state from orange cluster.  Absorbing state in red.  Top: diffusion time of 1, Bottom: diffusion time of 5.  $g=0.04$}\label{fig:absorbingState}
\end{figure}

We see in Figure \ref{fig:absorbingState} that, as time progresses, the state becomes more central to the Markov normalized process while still maintaining the geometry of the three clusters.  Moreover, the phase of the principal eigenvector converges to the page rank of the process, correctly placing the absorbing state in middling importance as in Figure \ref{fig:absorbingStatePageRank} in the appendix.
This is compared to the unnormalized magnetic Laplacian, which captures the absorbing nature of the red node, but loses the geometry of the other clusters or the relationship between the absorbing state and cluster 3. 

%

\section{Conclusion}
We note that Markov normalization of the graph adjacency matrix yields several benefits.  Normalization of the complex rotation matrix $T$ converges to the page rank of the graph for a large time, as in Section \ref{threecluster}.  The Markov normalization also stabilizes the phase embedding over a larger choice of $g$, as in most examples in Section \ref{examples}.  By emphasizing the absorbing states in the first eigenvector, this normalization separates the recurrent states into subsequent eigenvectors as in Section \ref{hiddenCircle}.  And finally, we can introduce diffusion time into the process and observe the multi-step neighborhood dynamics of the graph, as in most examples in Section \ref{examples}.


\newpage
\section*{Appendix}

\begin{figure}[!h]
\begin{tabular}{ccc}
\includegraphics[width=.3\textwidth]{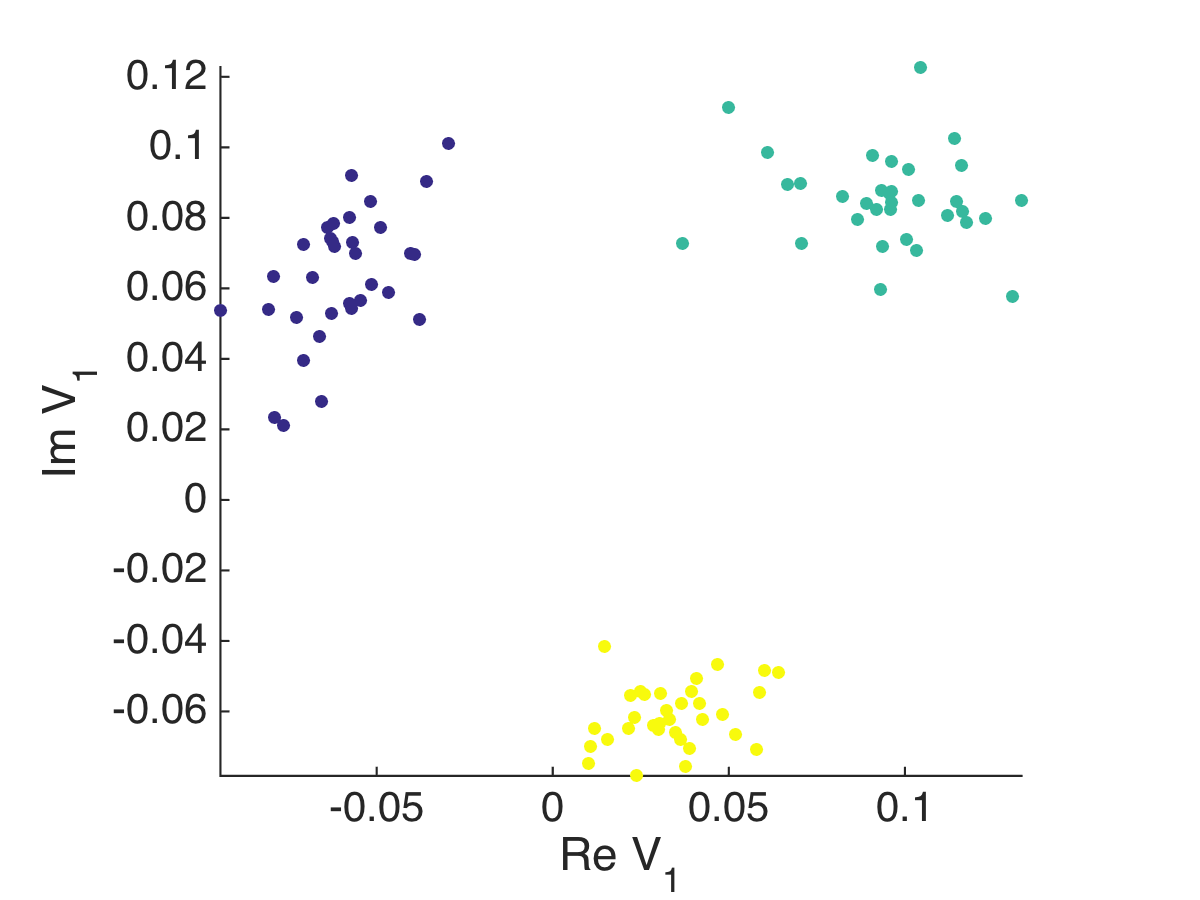} & 
\includegraphics[width=.3\textwidth]{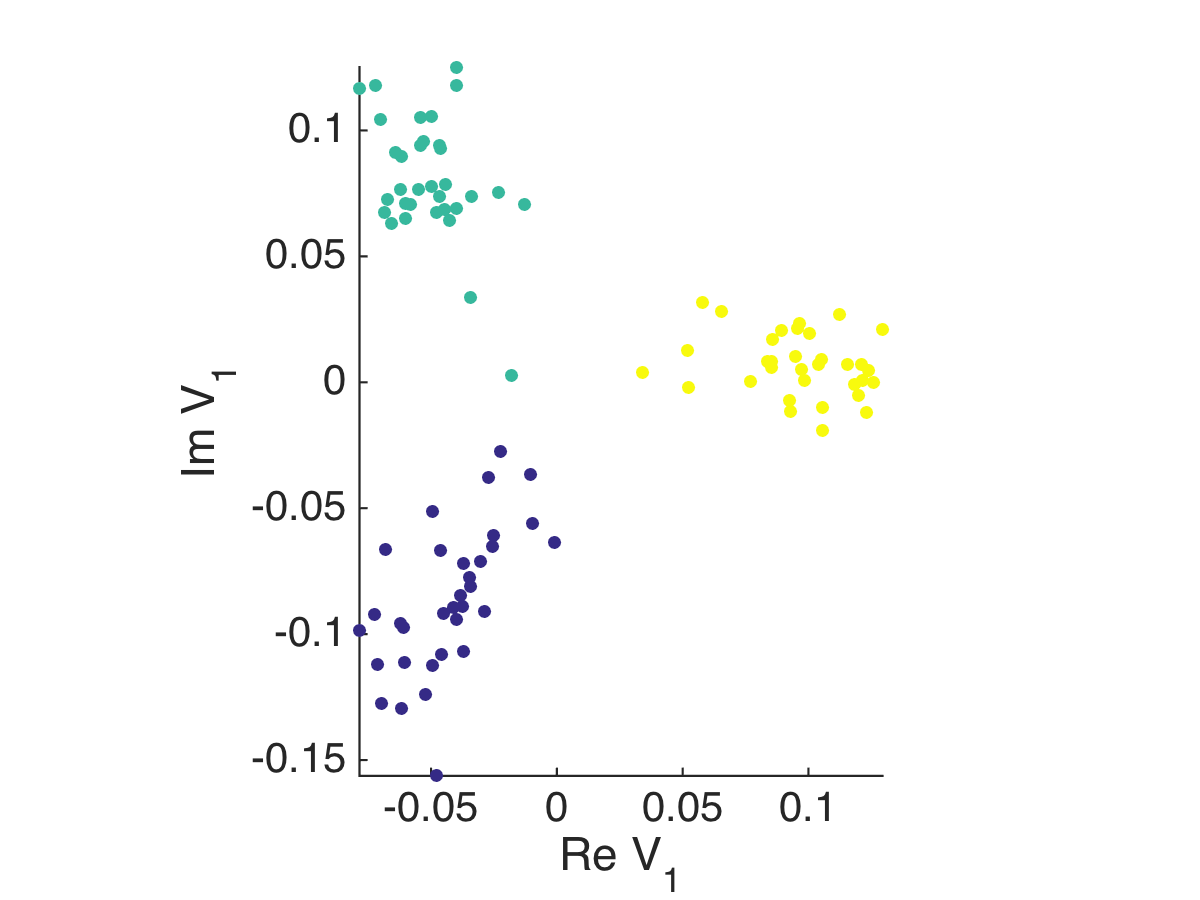} &
\includegraphics[width=.3\textwidth]{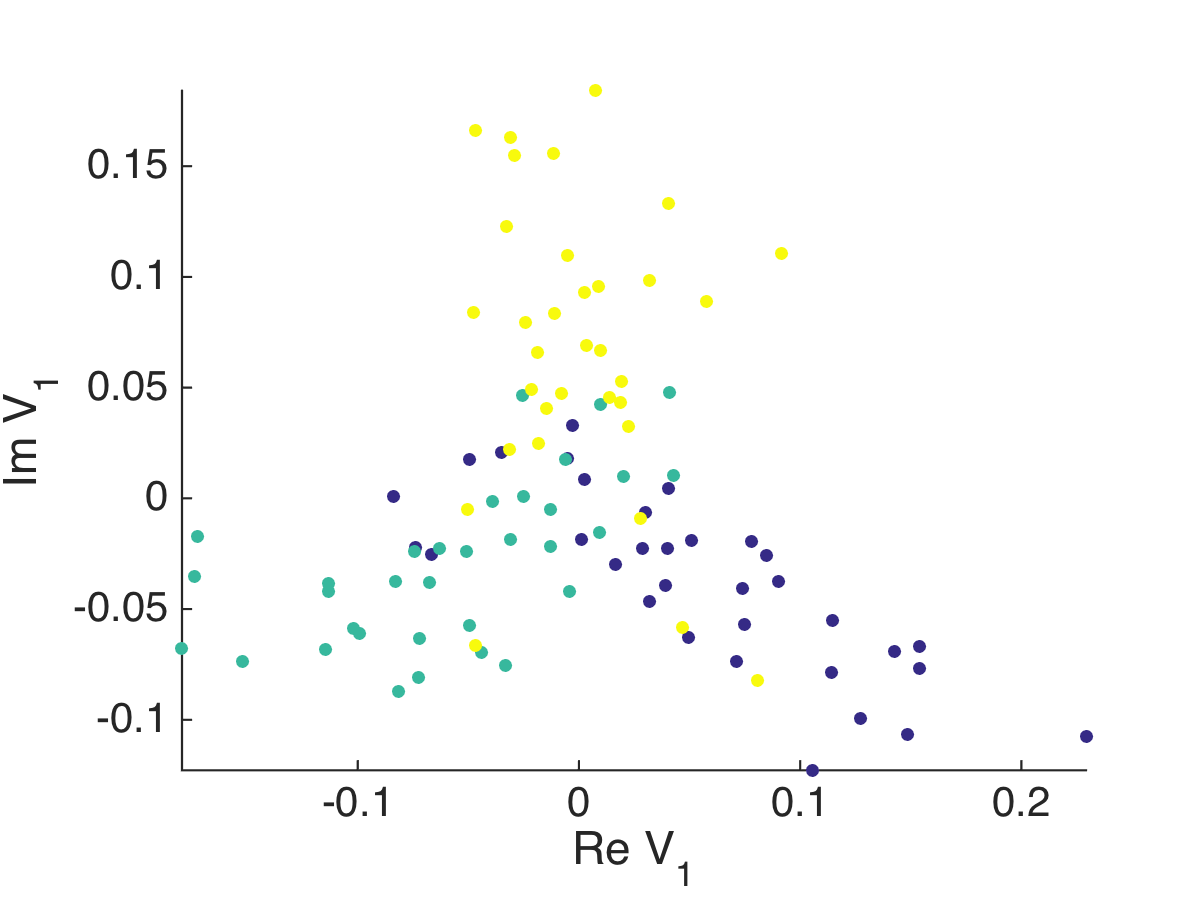} \\
t=1 & t=2 & t=3 \\
\includegraphics[width=.3\textwidth]{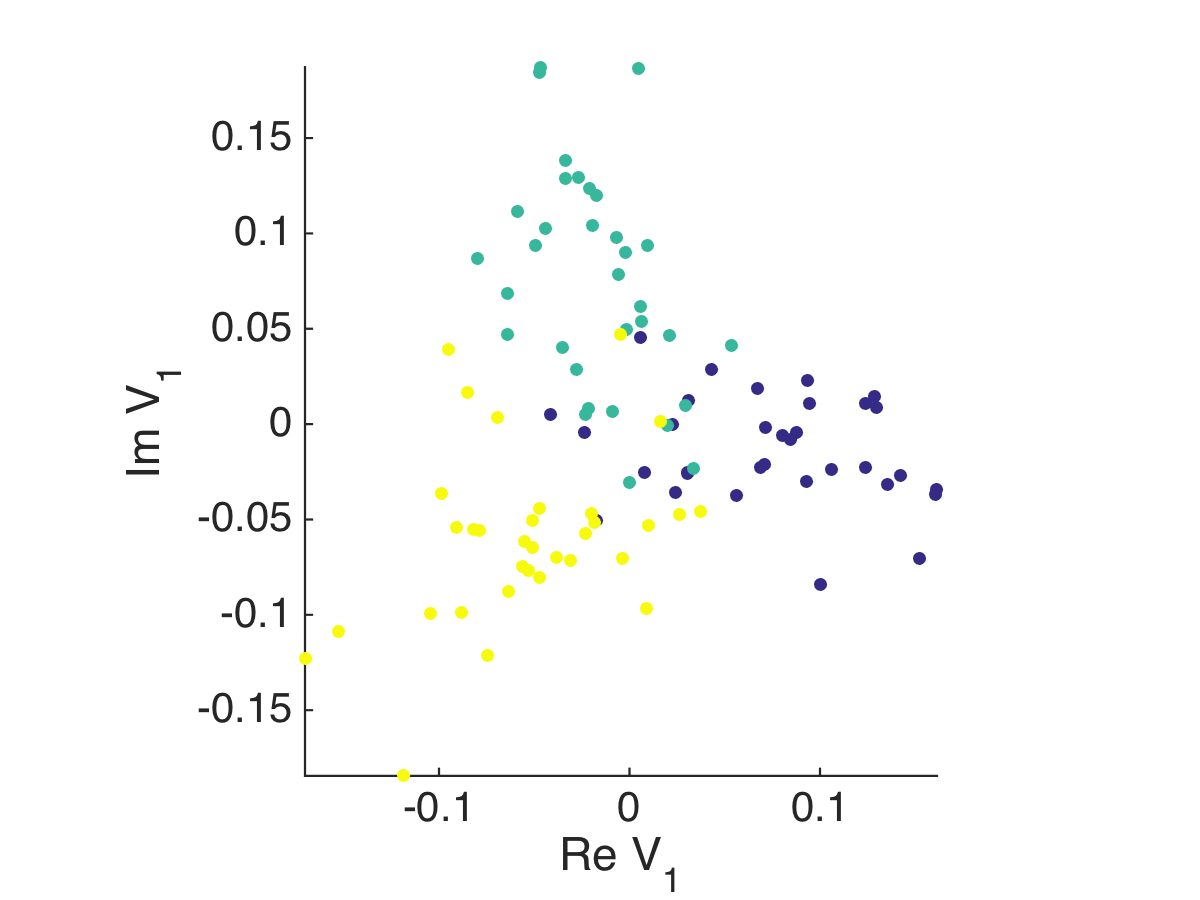} & 
\includegraphics[width=.3\textwidth]{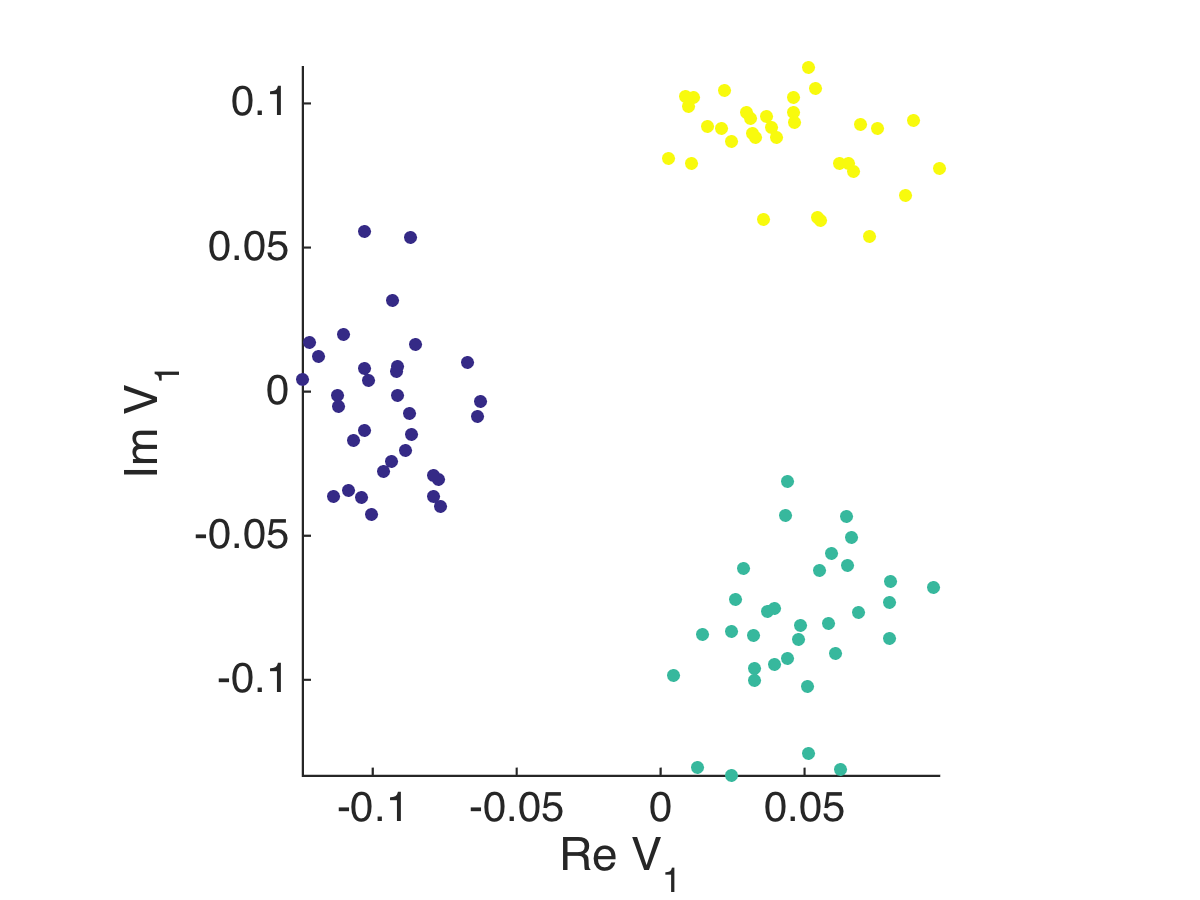} &
\includegraphics[width=.3\textwidth]{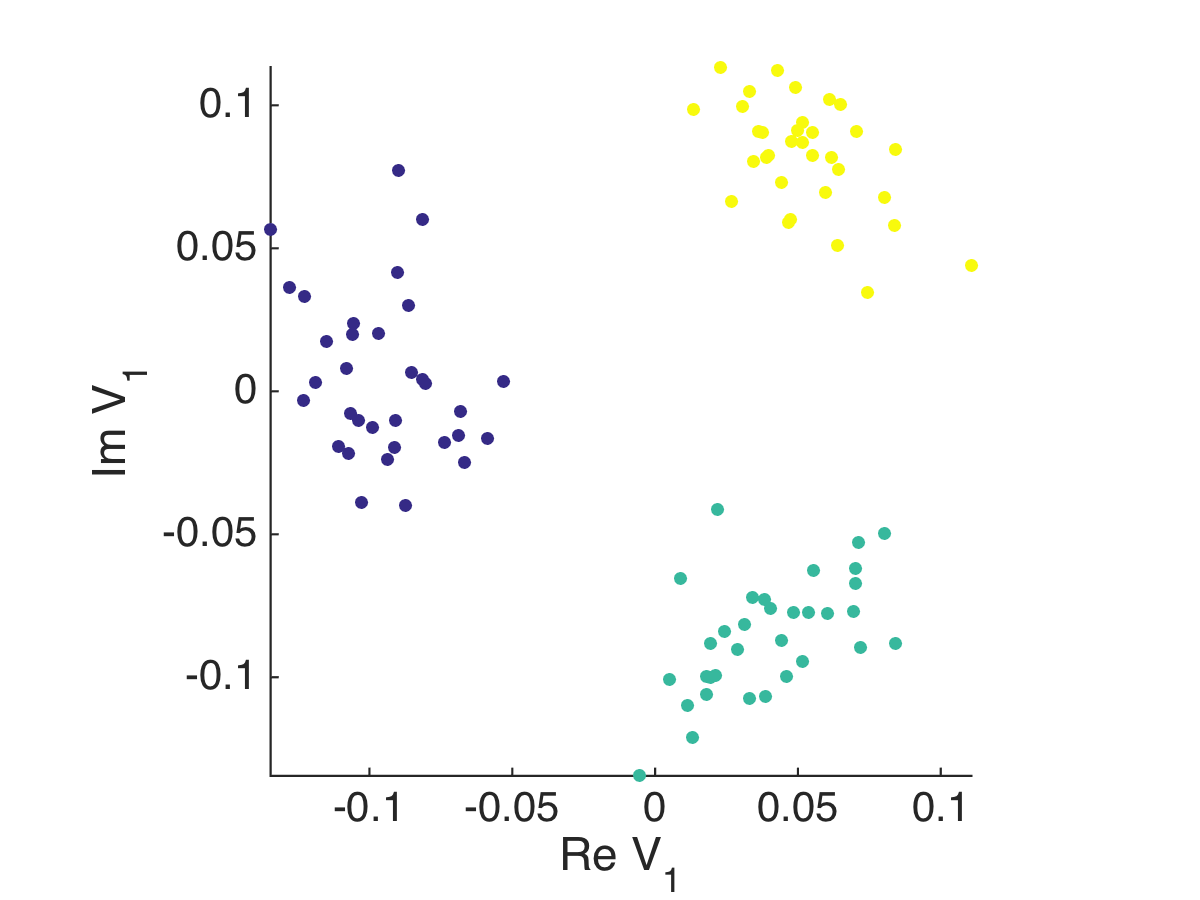} \\
t=4 & t=5 & t=6 \\
\includegraphics[width=.3\textwidth]{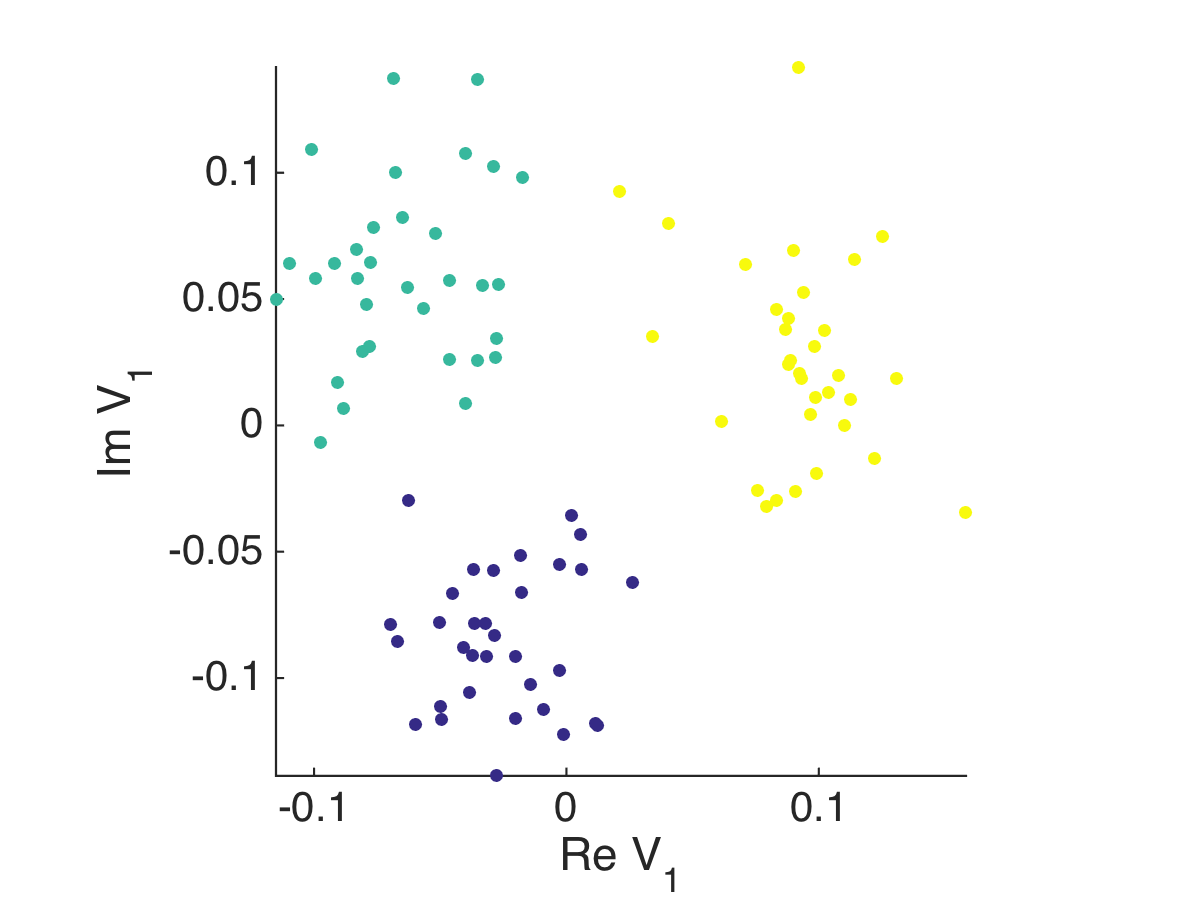} & 
\includegraphics[width=.3\textwidth]{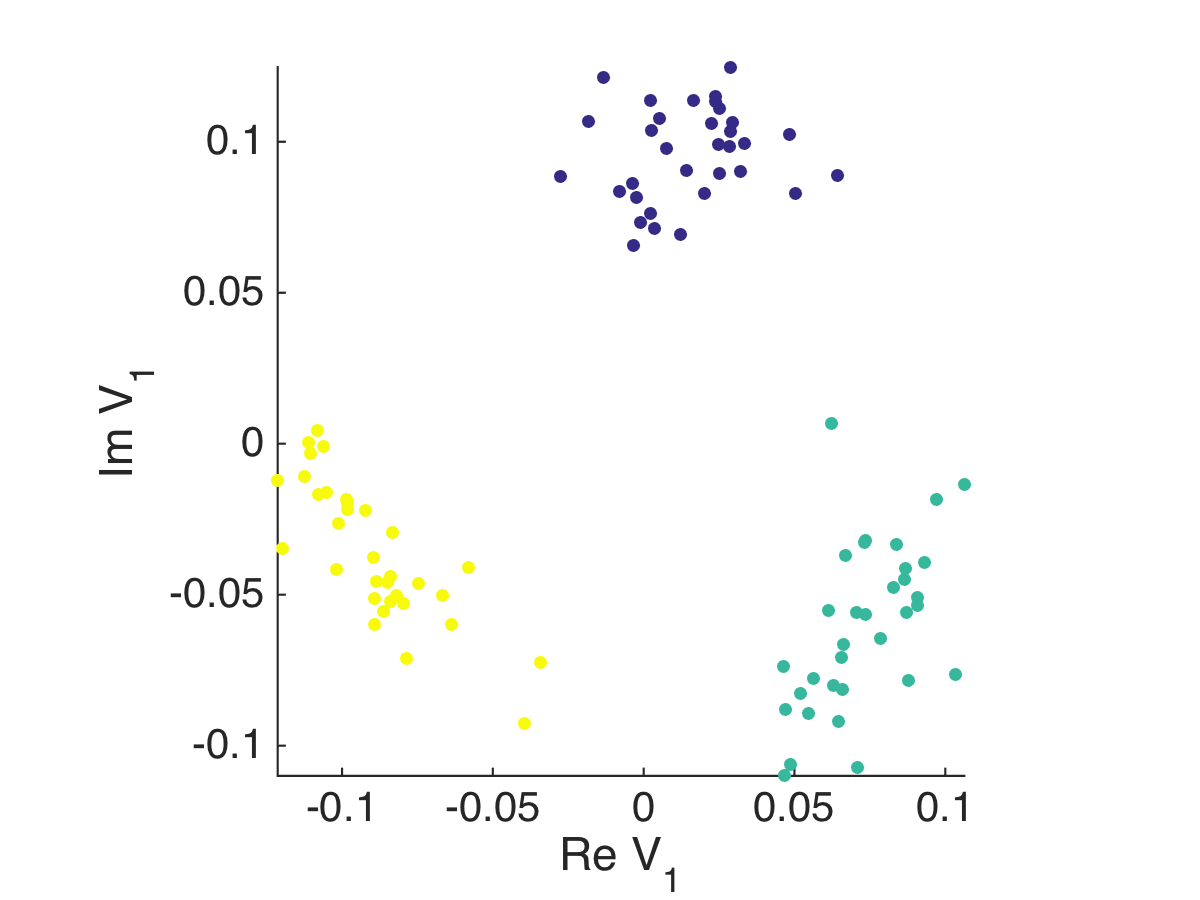} &
\includegraphics[width=.3\textwidth]{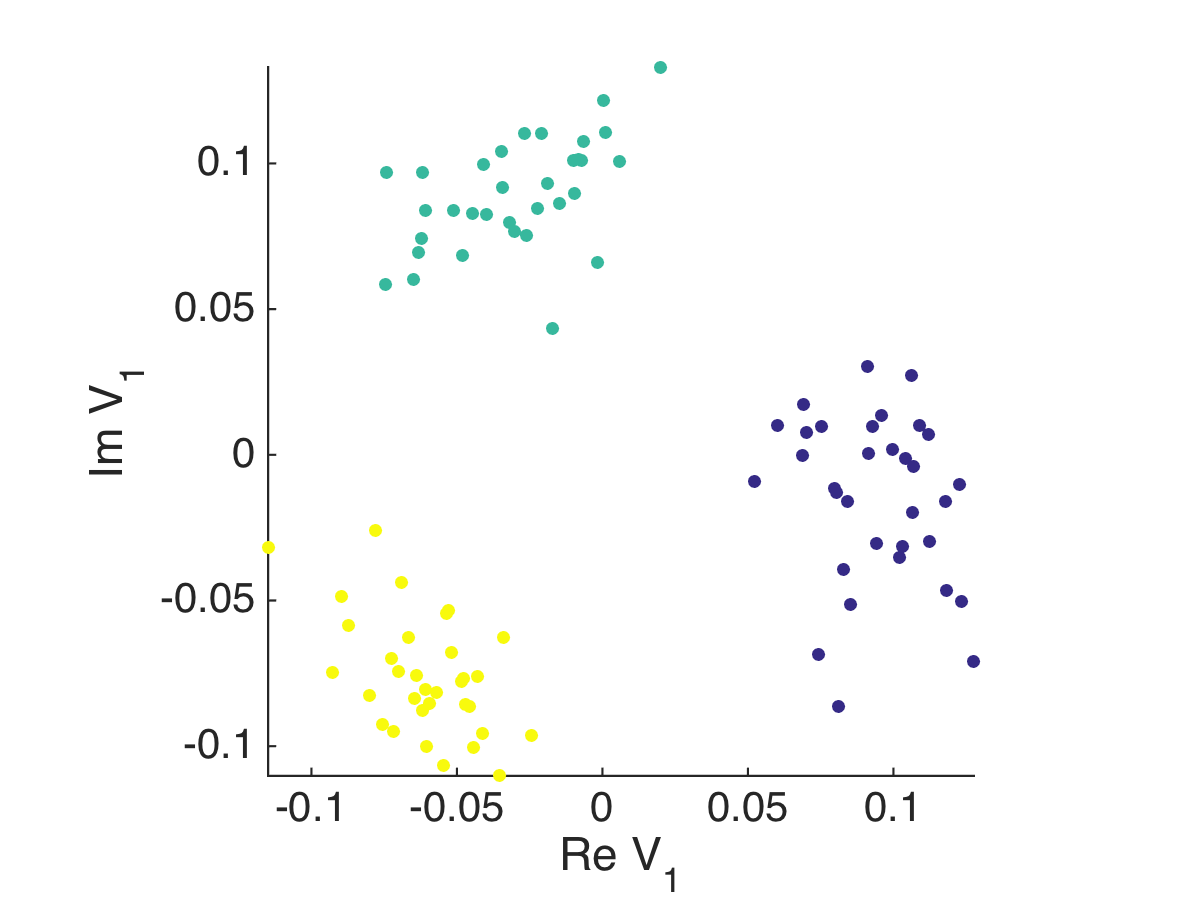} \\
t=7 & t=8 & t=9 \\
\end{tabular}
\caption{Three clusters with $P(inCluster)=0.5$, $P(outCluster)=0.5$, $P(rotateClockwise) = 0.9$, and with $g=1/4$.  Evolve cluster from $t\in\{1,...,9\}$.}\label{fig:theeClusterSpin}
\end{figure}

\begin{figure}[!h]
\begin{center}
\includegraphics[width=.4\textwidth]{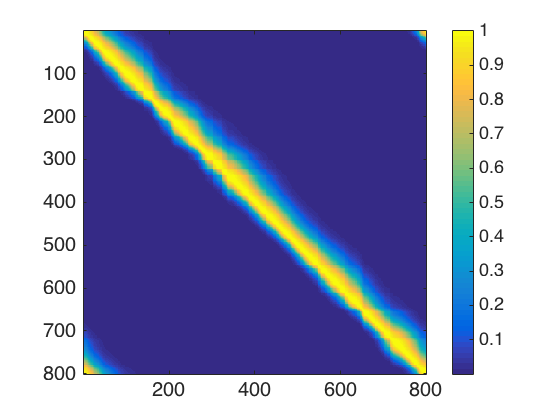}
\end{center}
\caption{Affinity matrix for circle with drift}\label{fig:circleDriftAff}
\end{figure}

\begin{figure}[!h]
\begin{tabular}{ccc}
\includegraphics[width=.3\textwidth]{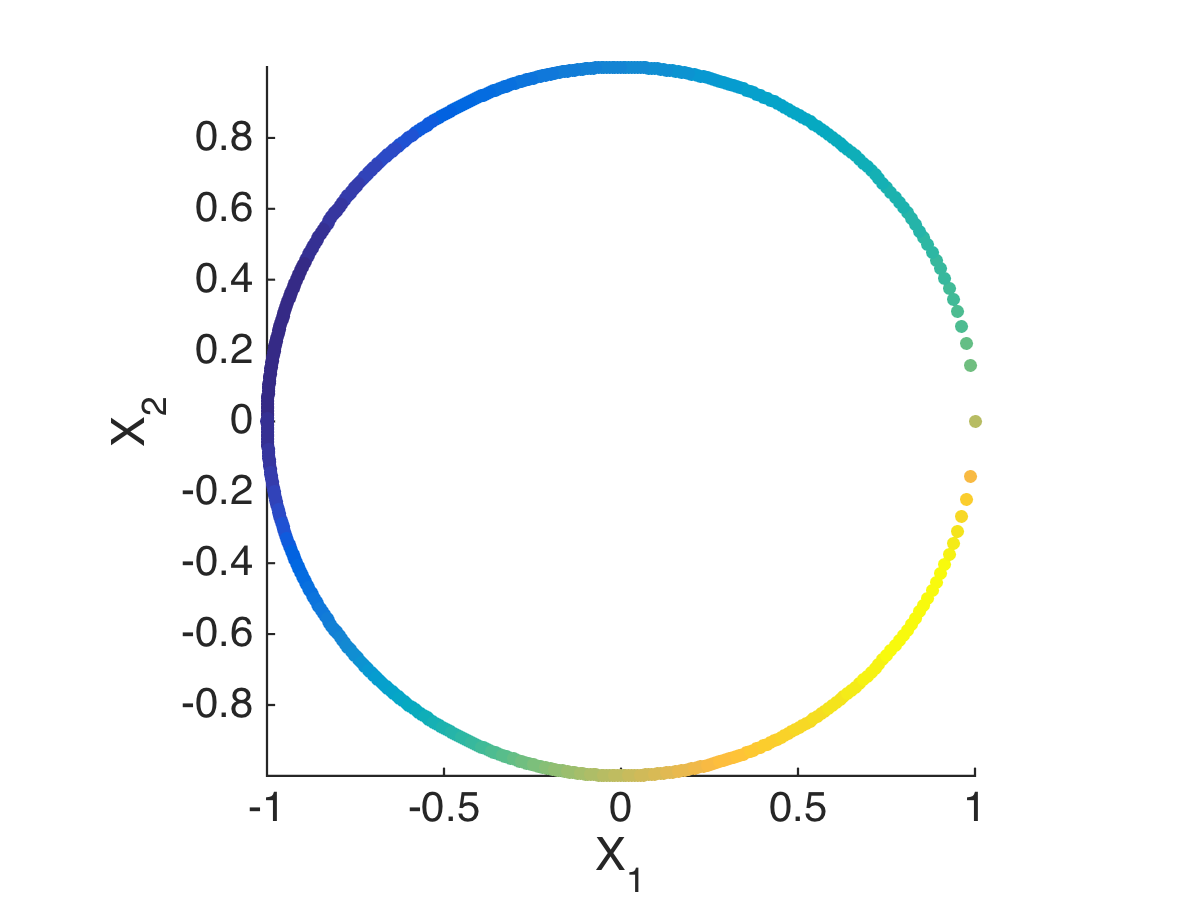} & 
\includegraphics[width=.3\textwidth]{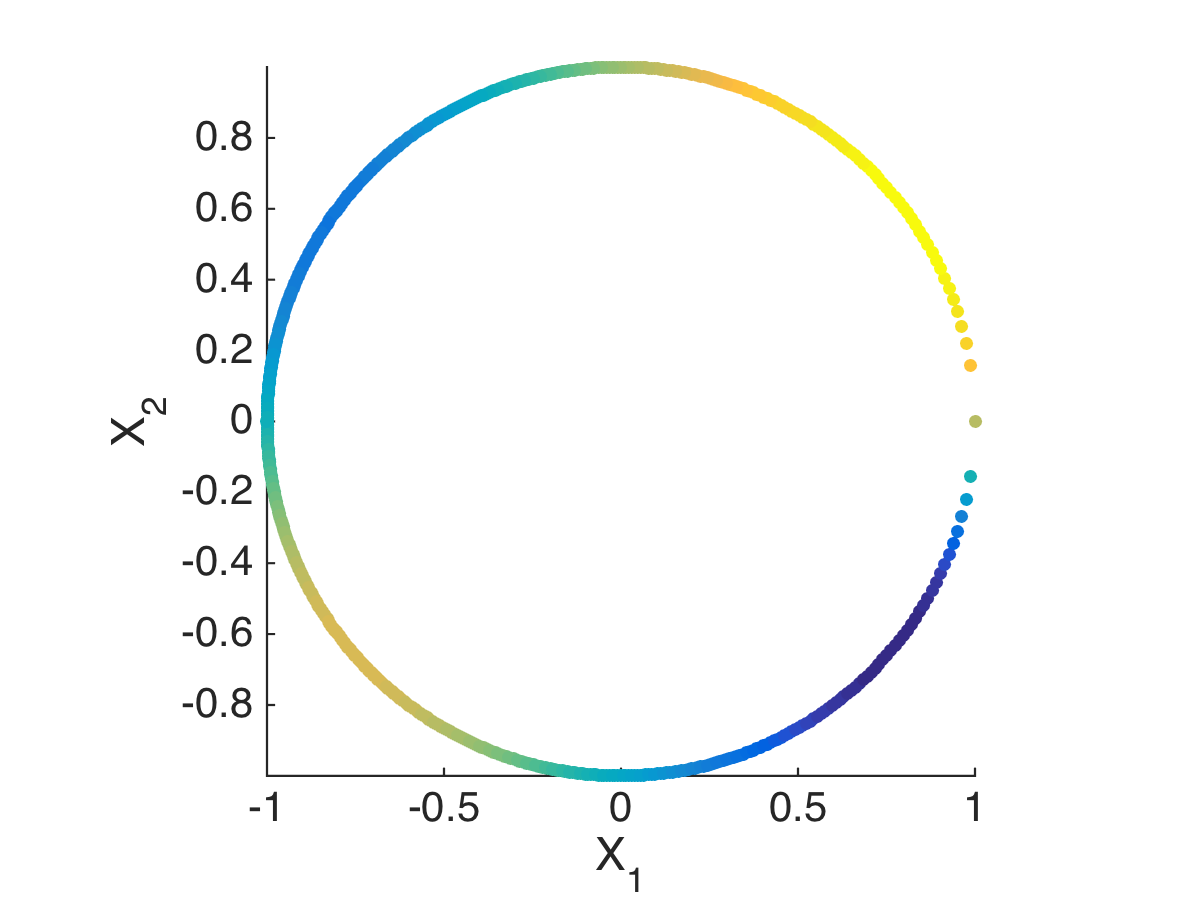} &
\includegraphics[width=.3\textwidth]{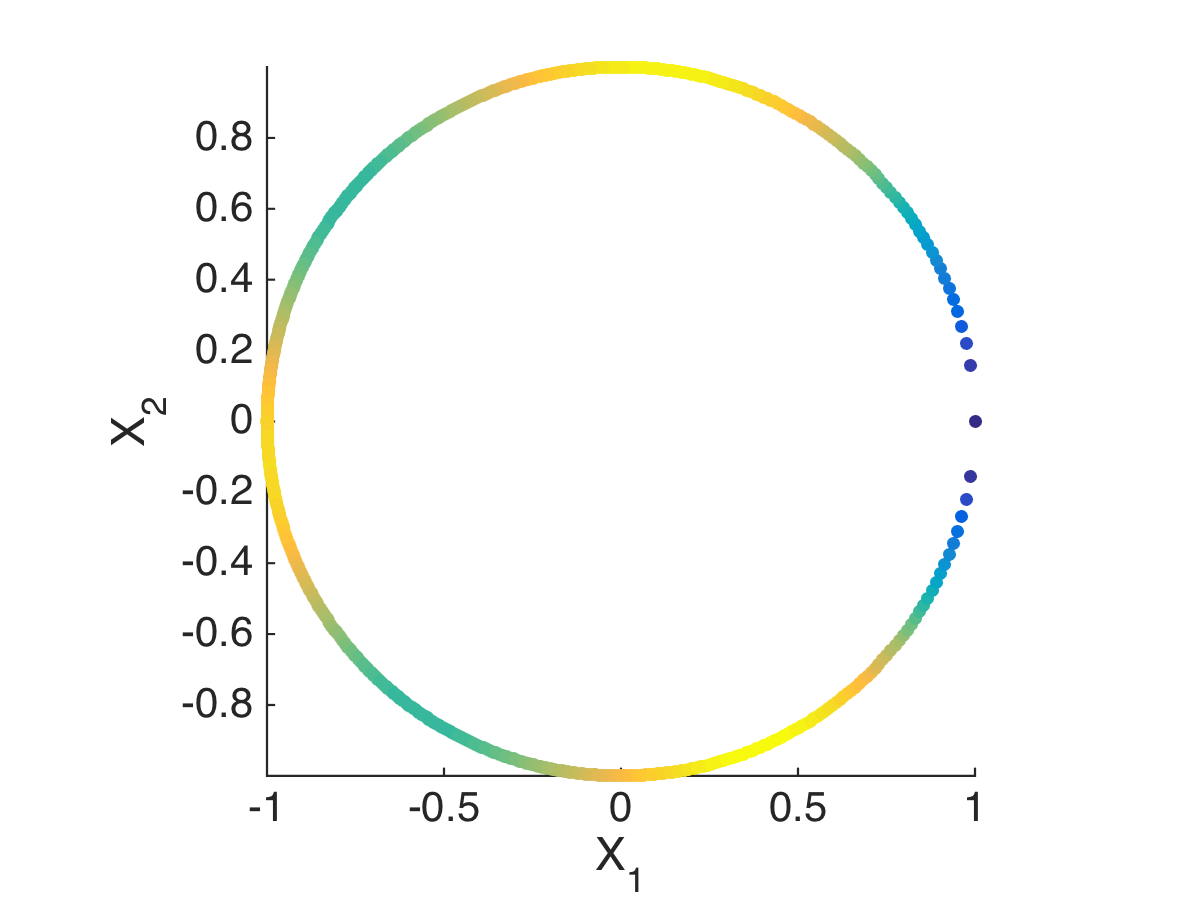} \\
Normalized real$(\phi_1)$ & Normalized  real$(\phi_3)$ &Normalized  real$(\phi_5)$\\
\includegraphics[width=.3\textwidth]{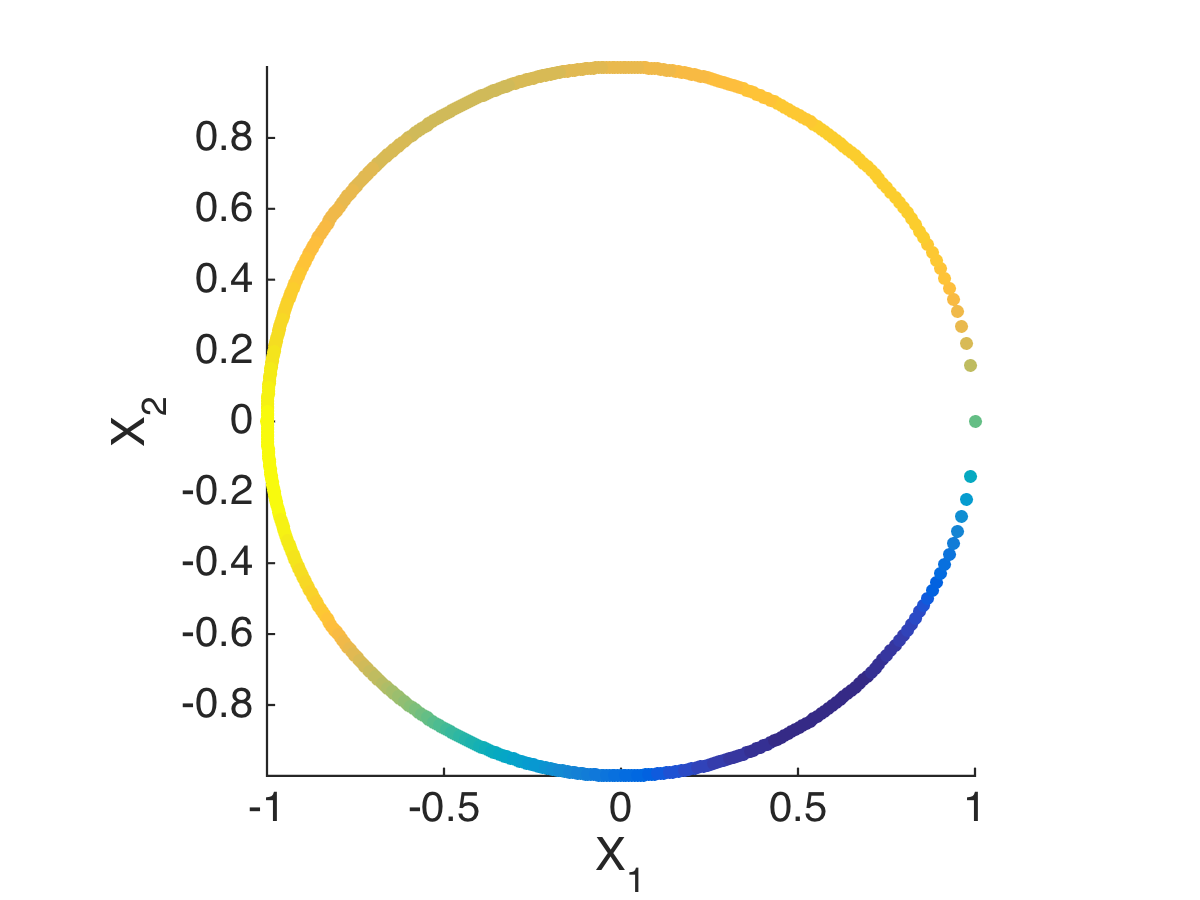} & 
\includegraphics[width=.3\textwidth]{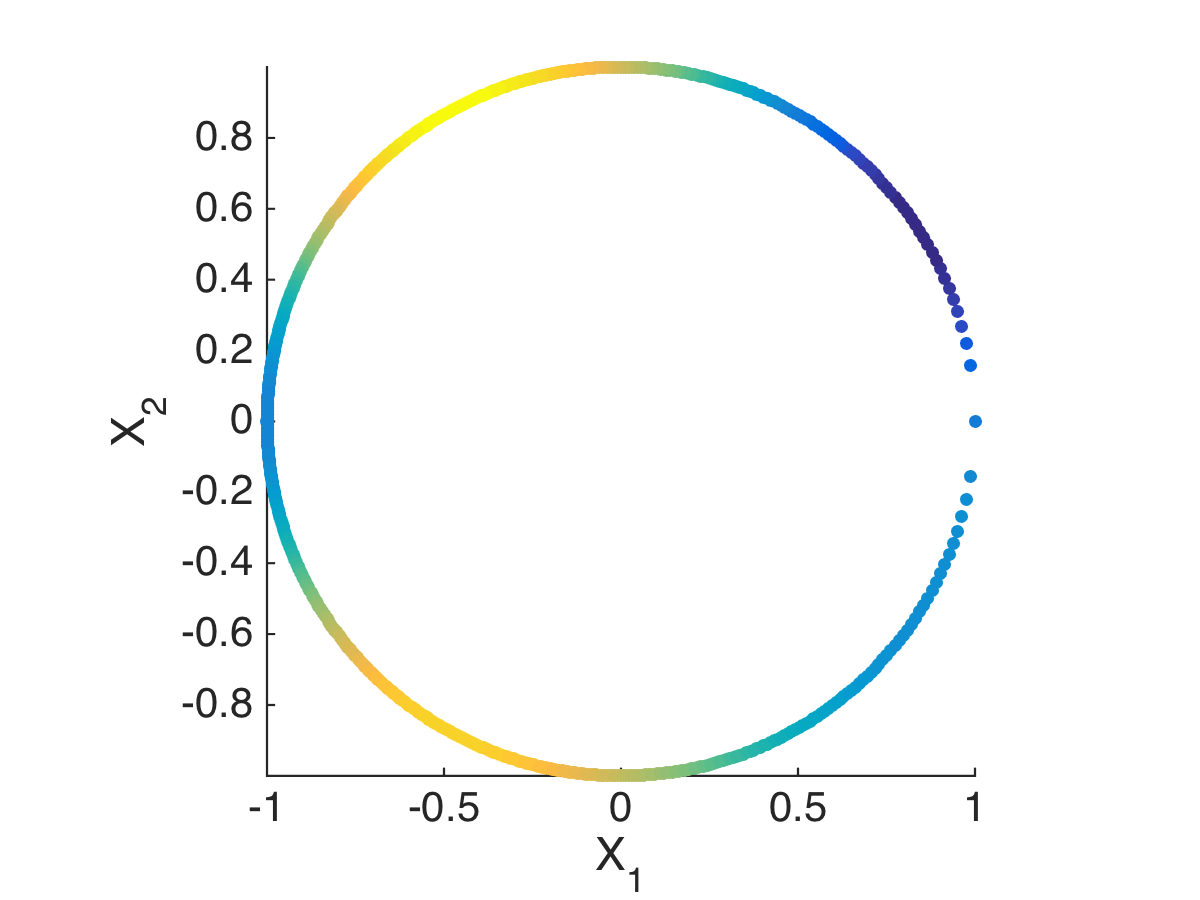} &
\includegraphics[width=.3\textwidth]{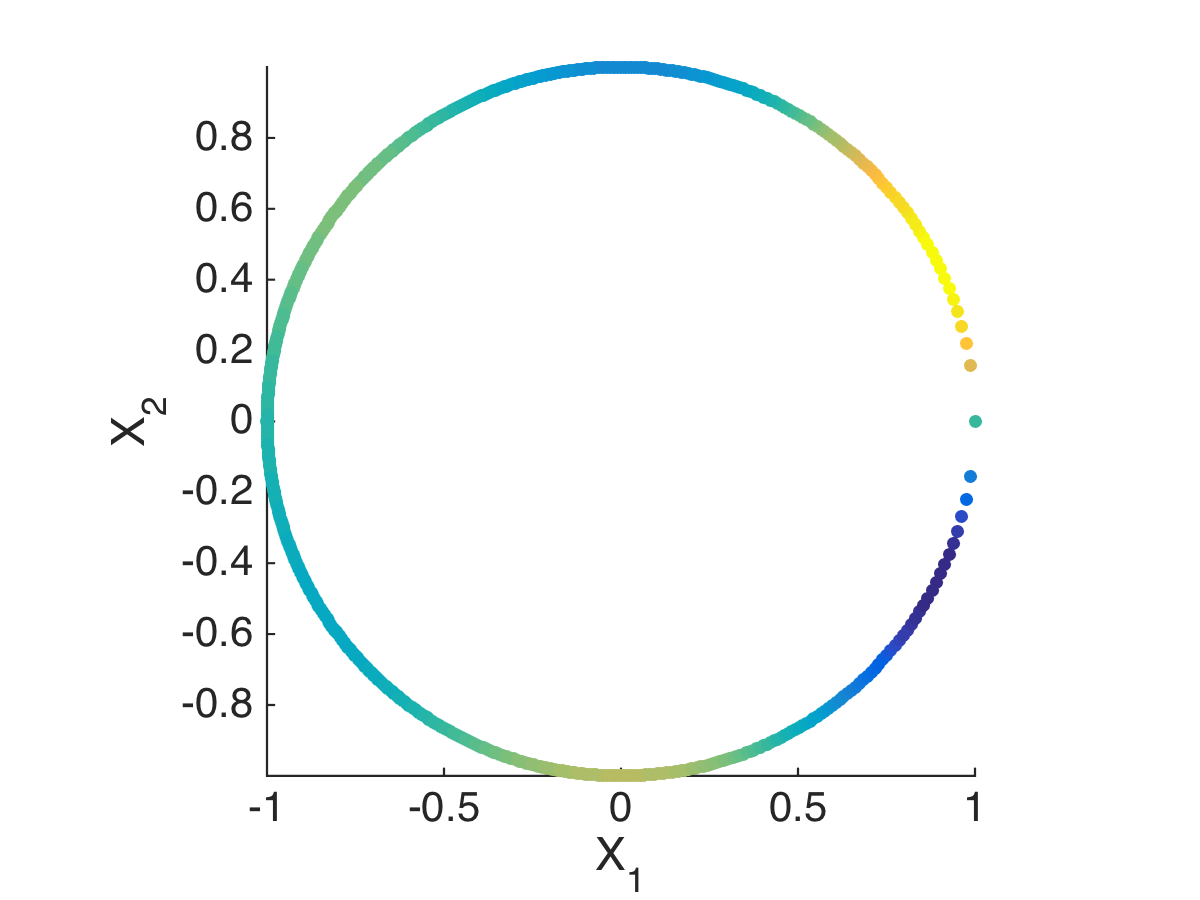} \\
Unnormalized real$(\phi_1)$ &Unnormalized real$(\phi_3)$ &Unnormalized real$(\phi_5)$\\
\end{tabular}
\caption{Circle with drift, colored by corresponding real part of real$(\phi_i)$.  Top: Markov normalized, Bottom: Unnormalized.  The eigenvectors correspond to multiples of sinusoids, the eigenvectors of the DFT matrix.  }\label{fig:circleDriftEigs}
\end{figure}

\begin{figure}[!h]
\begin{tabular}{ccc}
\includegraphics[width=.3\textwidth]{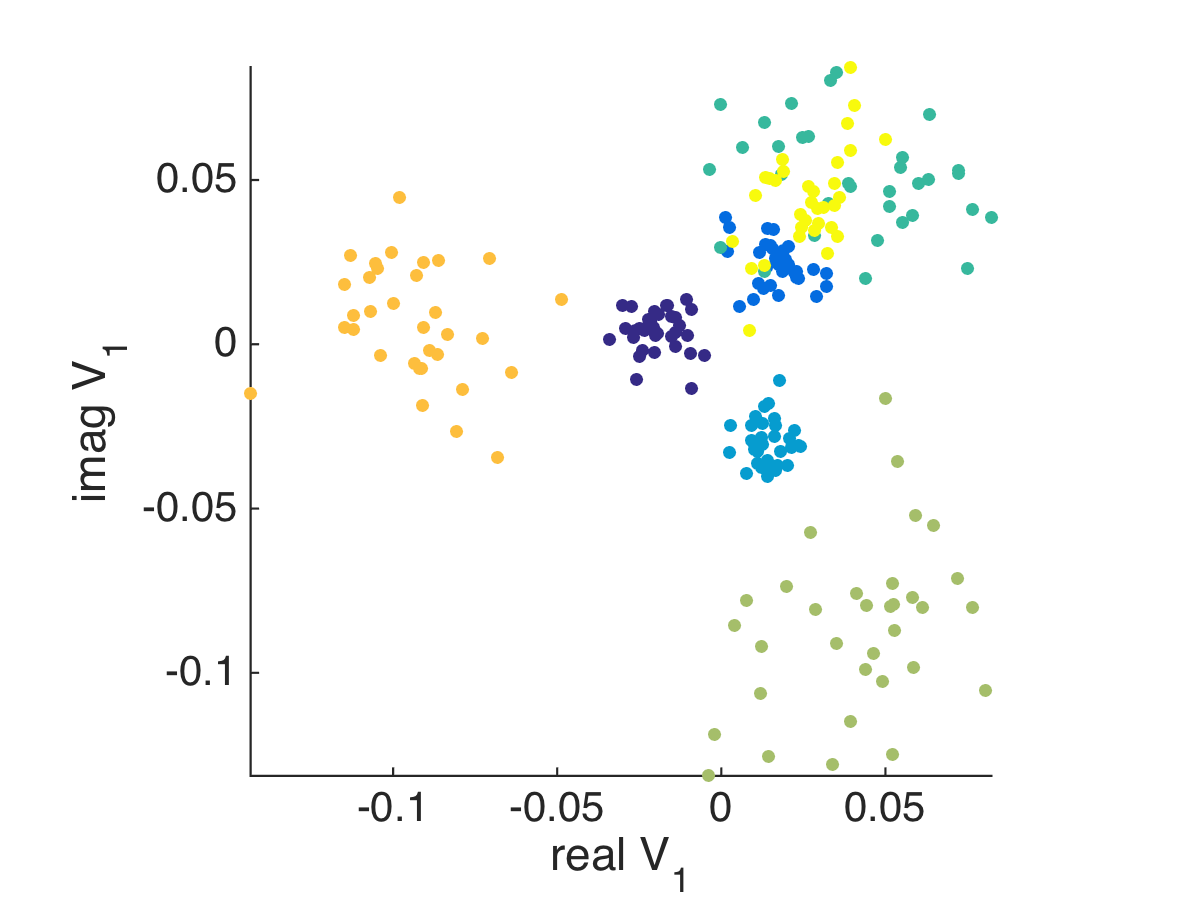} & 
\includegraphics[width=.3\textwidth]{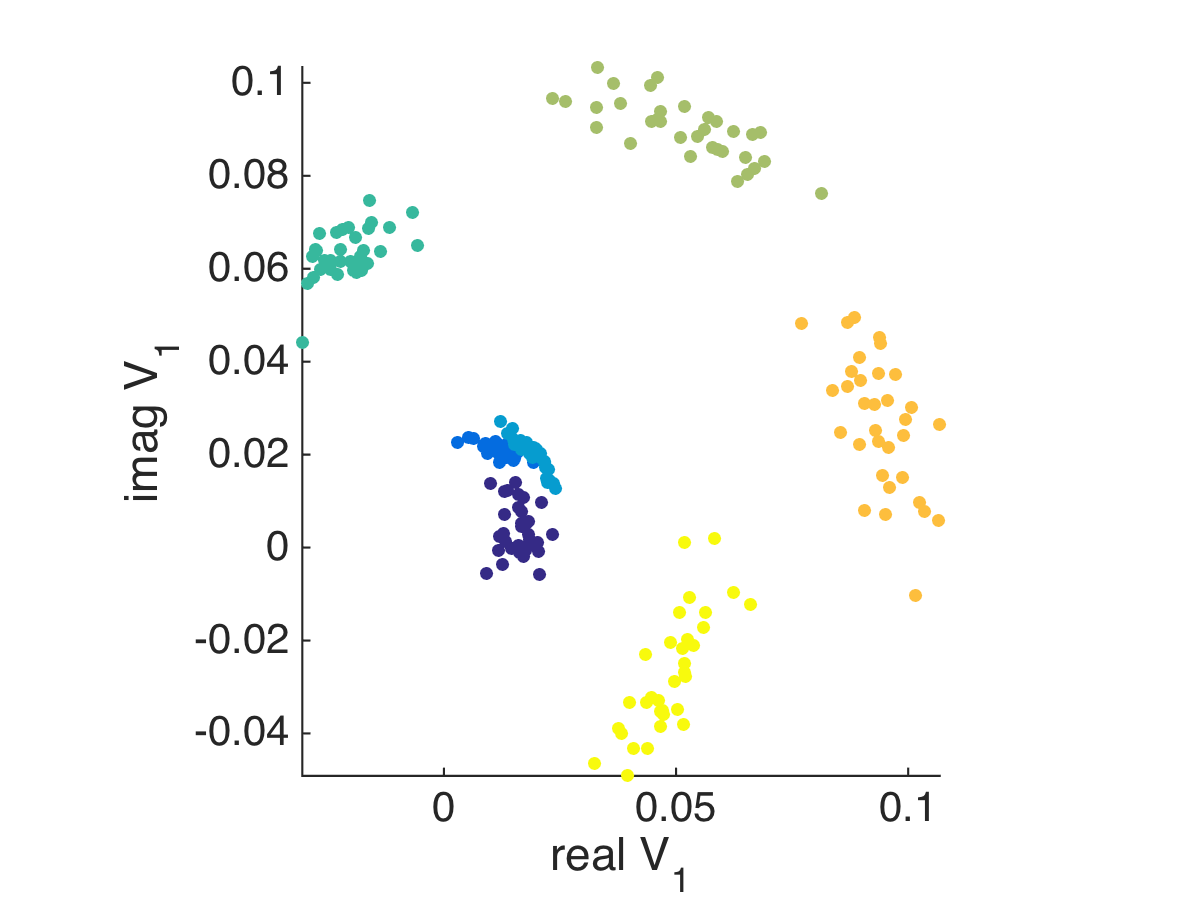} &
\includegraphics[width=.3\textwidth]{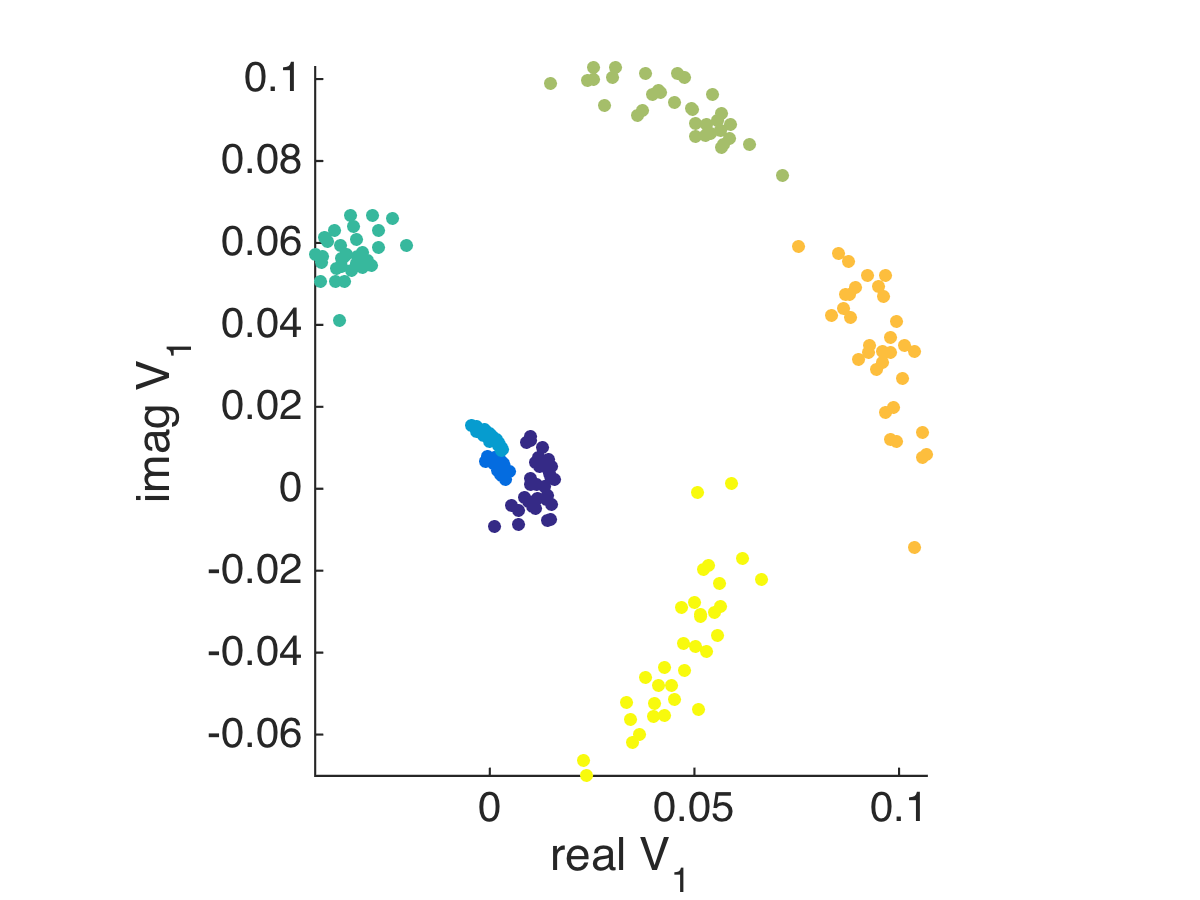} \\
t=1 & t=2 & t=3 \\
\includegraphics[width=.3\textwidth]{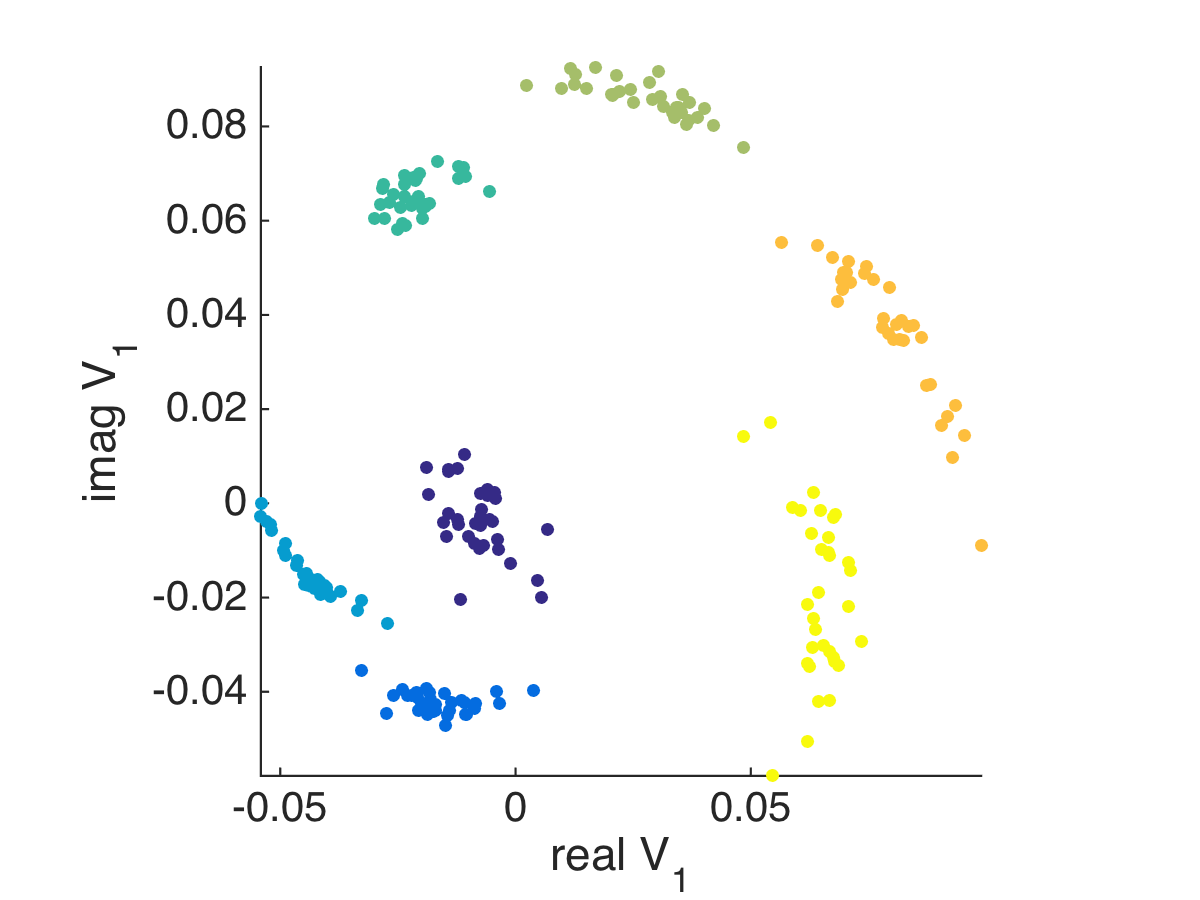} & 
\includegraphics[width=.3\textwidth]{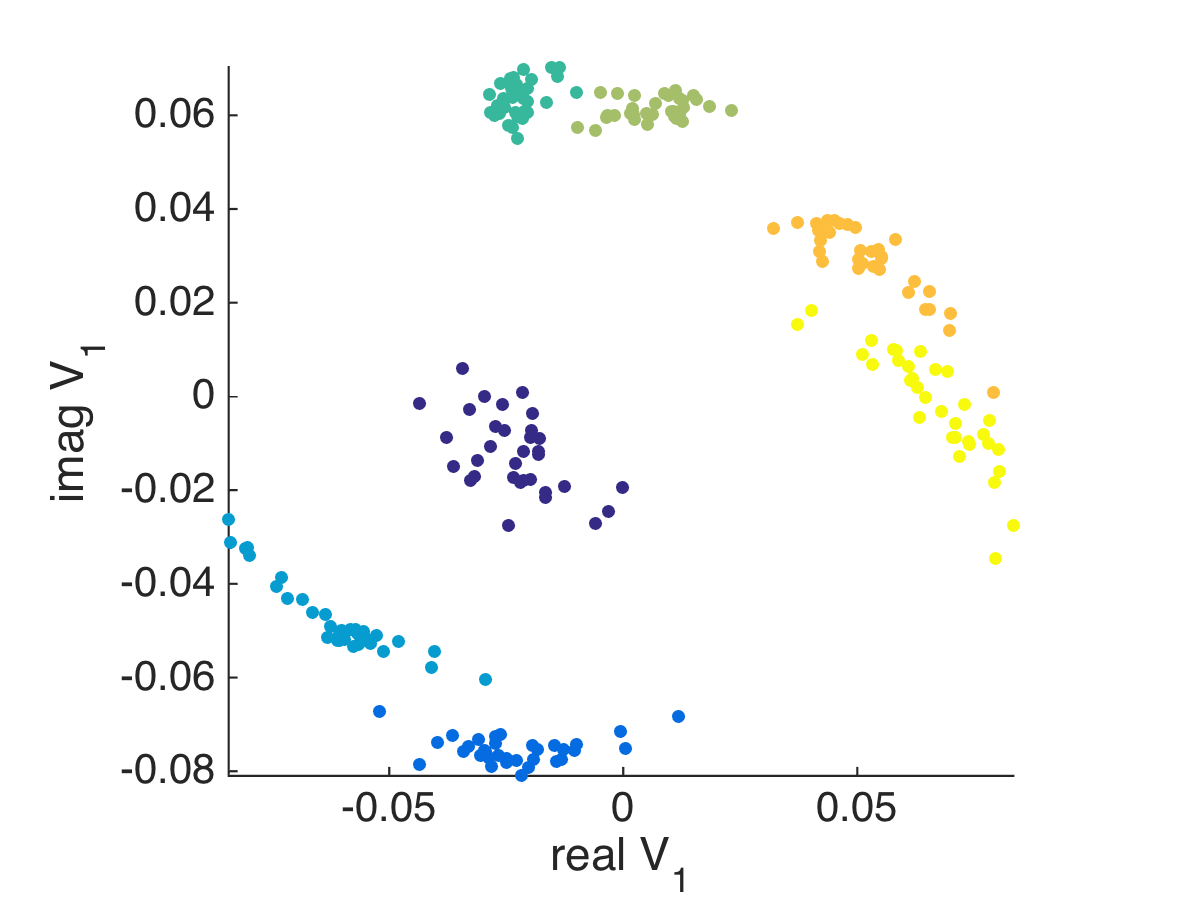} &
\includegraphics[width=.3\textwidth]{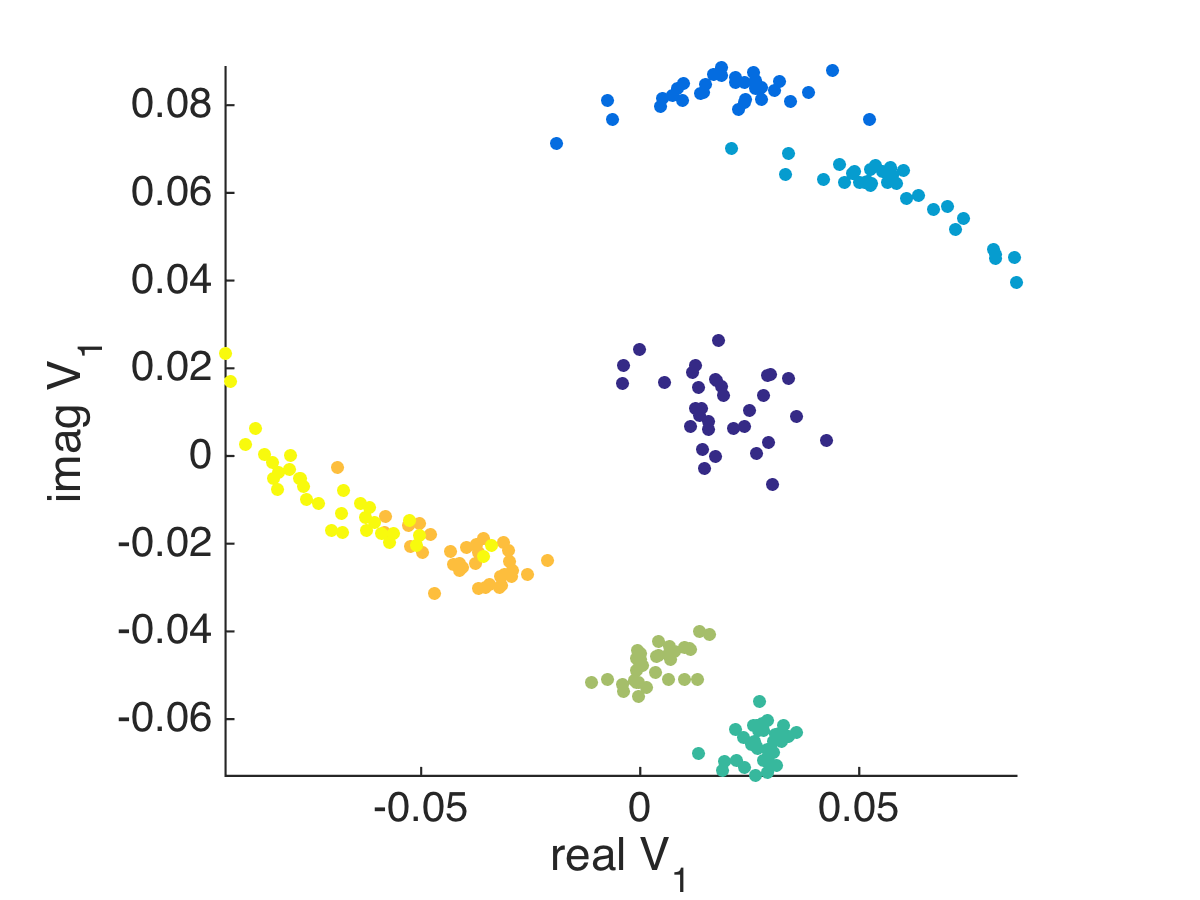} \\
t=4 & t=5 & t=6 \\
\includegraphics[width=.3\textwidth]{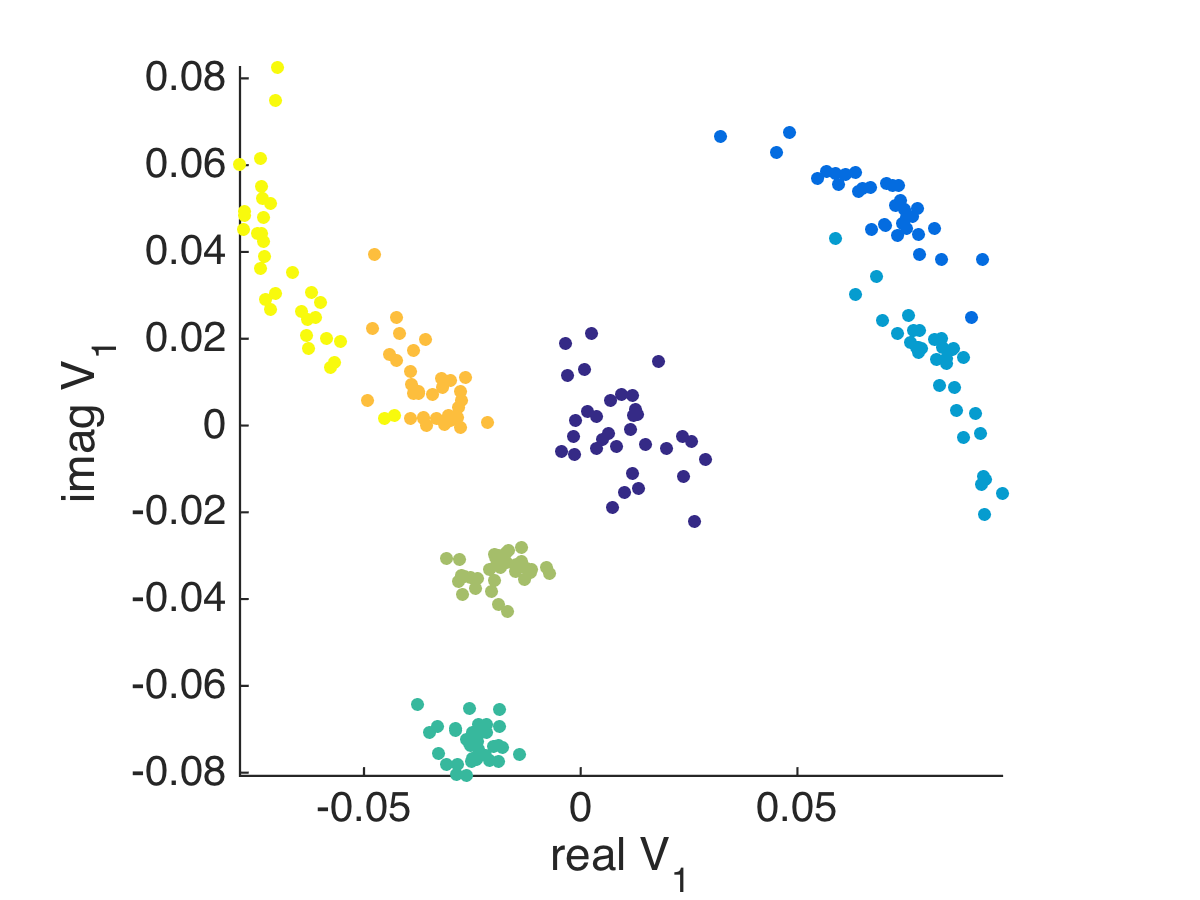} & 
\includegraphics[width=.3\textwidth]{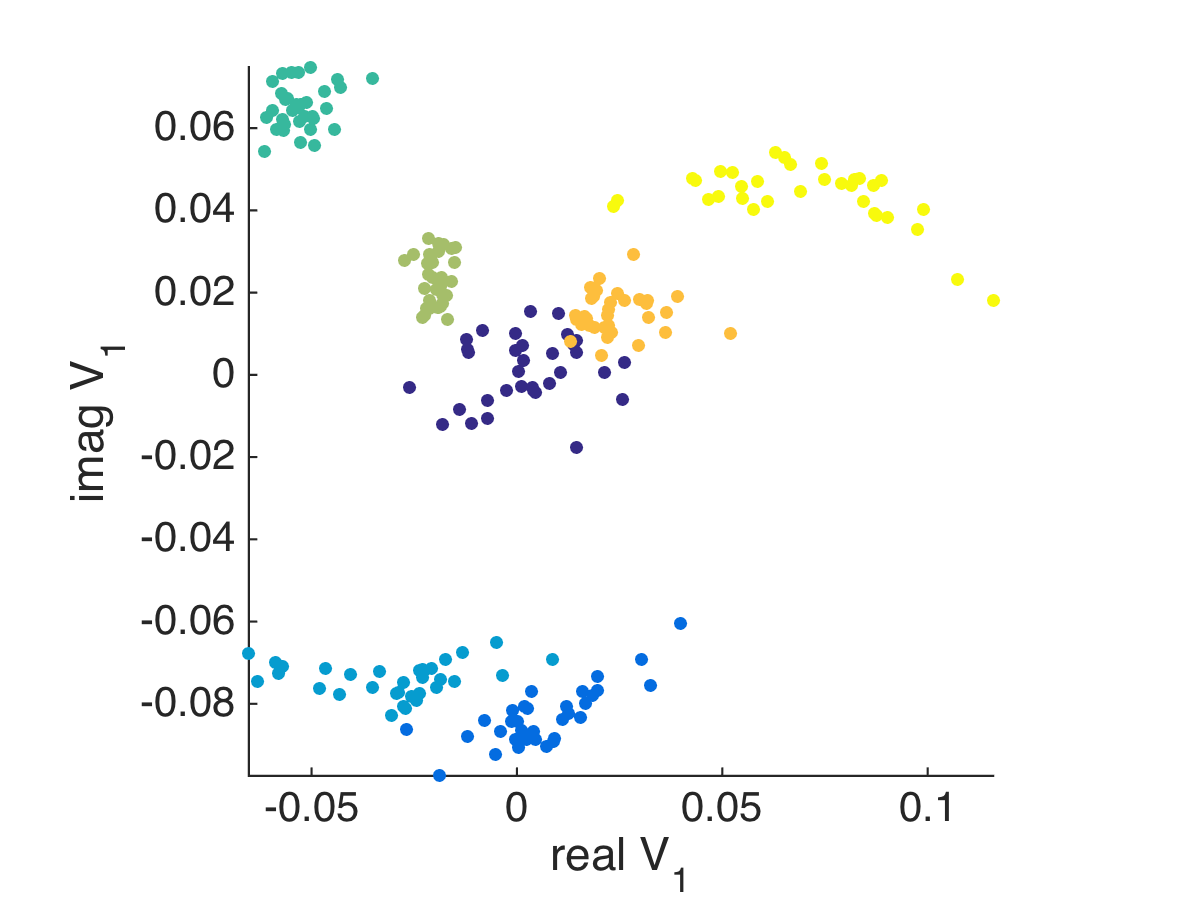} &
\includegraphics[width=.3\textwidth]{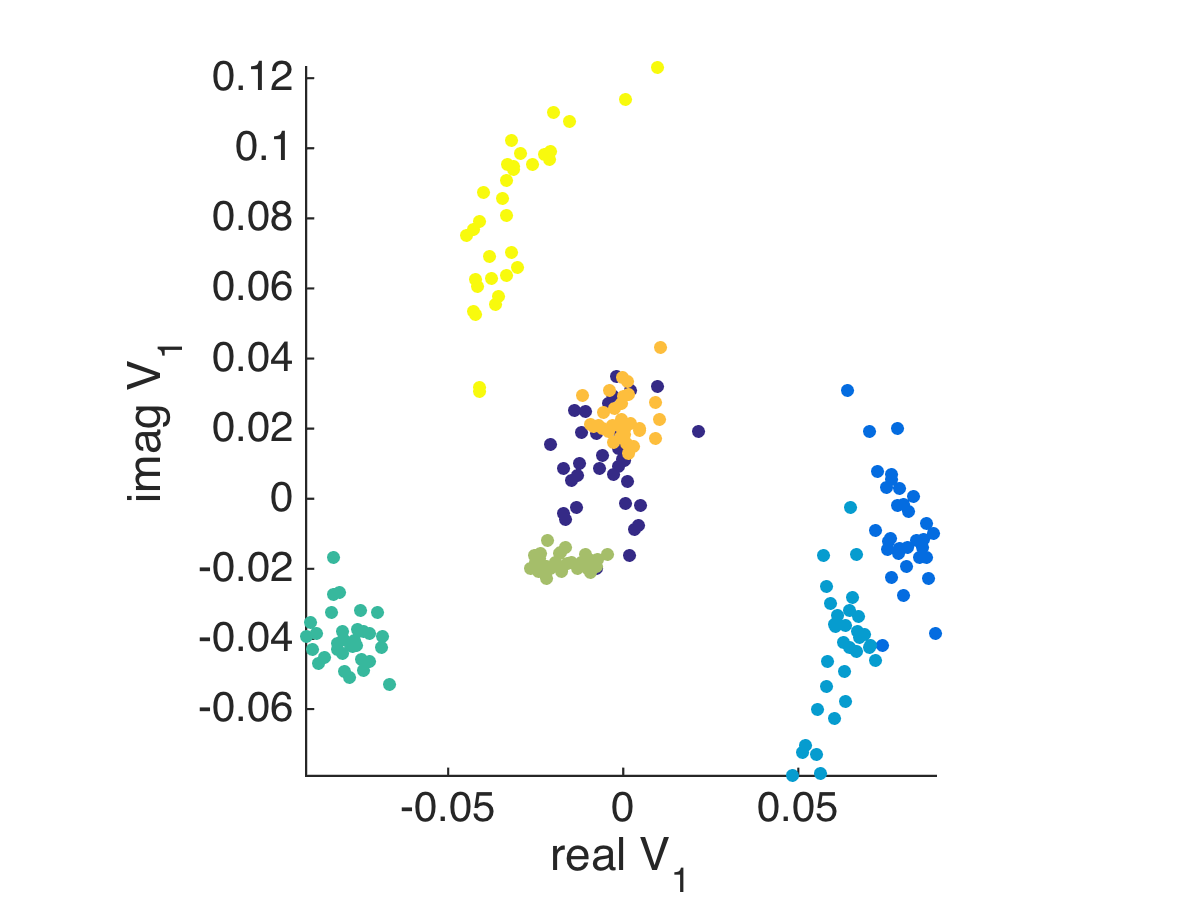} \\
t=7 & t=8 & t=9 \\
\end{tabular}
\caption{Bow tie clusters with $P(inCluster)=0.5$, $P(outCluster)=0.5$, $P(rotateClockwise) = 0.9$, and with $g=1/4$.  Evolve cluster from $t\in\{1,...,9\}$.}\label{fig:bowtieSpin}
\end{figure}

\begin{figure}[!h]
\begin{tabular}{cc}
\includegraphics[width=.3\textwidth]{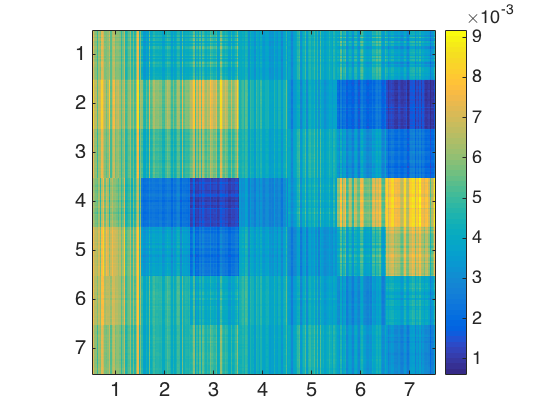} & 
\includegraphics[width=.3\textwidth]{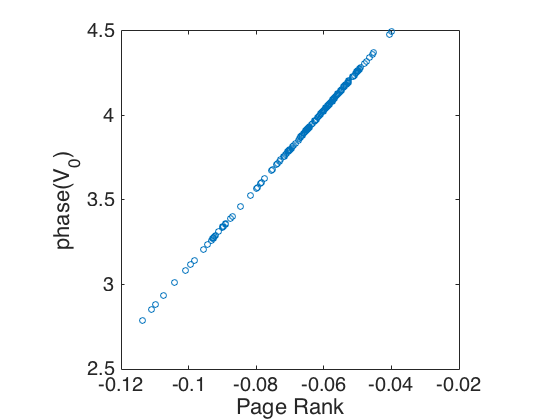} 
\end{tabular}
\caption{Bow tie clusters affinity at $t=7$ (left), and bow tie clusters phase$(V_0)$ at $t=10$ (right).}\label{fig:bowtieAff}
\end{figure}

\begin{figure}[!h]
\begin{center}
\includegraphics[width=.4\textwidth]{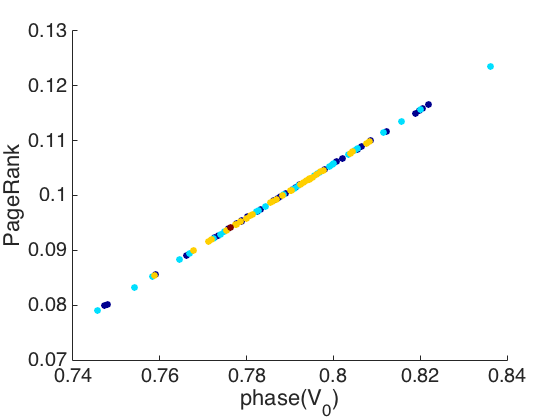} 
\end{center}
\caption{Absorbing state principal eigenvector phase compared to the page rank of the process.}\label{fig:absorbingStatePageRank}
\end{figure}

\end{document}